\documentclass[conference,letterpaper]{IEEEtran}

\usepackage[utf8]{inputenc} 
\usepackage[T1]{fontenc}
\usepackage{url}
\usepackage{ifthen}
\usepackage{cite}
\usepackage[cmex10]{amsmath} 

\usepackage{tikz}
\usepackage{tcolorbox}
\usepackage{algorithm}
\usepackage{algpseudocode}
\usepackage[colorlinks=true,linkcolor=black,urlcolor=blue!60!black,citecolor=black]{hyperref}

\usetikzlibrary{arrows.meta,positioning}
\tcbuselibrary{skins,breakable}

\definecolor{boxbg}{RGB}{246,246,244}
\definecolor{boxln}{RGB}{200,200,196}
\definecolor{accent}{RGB}{40,90,170}

\newcommand{\Ftwo}{\mathbb{F}_2}
\newcommand{\Acc}{\Sigma}

\newcommand{\supp}{\operatorname{supp}}

\newcommand{\bh}{\mathbf{h}}
\newcommand{\zero}{\mathbf{0}}
\newcommand{\by}{\mathbf{y}}
\newcommand{\FF}{\mathbb{F}}

\newcommand{\bc}{\mathbf{c}}

\newcommand{\bx}{\mathbf{x}}
\newcommand{\be}{\mathbf{e}}
\newcommand{\xor}{\oplus}

\newcommand{\Int}{\mathcal{I}}
\DeclareMathOperator{\dlog}{dlog}

\newtcolorbox{keybox}{
  colback=boxbg, colframe=boxln, boxrule=0.4pt, arc=2pt,
  left=7pt, right=7pt, top=6pt, bottom=6pt, breakable
}

\usepackage{stmaryrd}
\usepackage{mathtools}
\usepackage{amssymb}
\usepackage{tikz}
\usepackage{subfig}
\usepackage{amsthm}
\usepackage{pgfplots}
\usepackage{soul}
\usepackage{booktabs}

\theoremstyle{theorem}
\newtheorem{definition}{Definition}
\newtheorem{example}{Example}
\newtheorem{proposition}{Proposition}
\newtheorem{theorem}{Theorem}
\newtheorem{lemma}{Lemma}
\newtheorem{corollary}{Corollary}
\newtheorem{claim}{Claim}
\theoremstyle{remark}
\newtheorem{remark}{Remark}

\newcommand{\mc}{\mathcal}
\newcommand{\mbf}{\mathbf}
\newcommand{\mbb}{\mathbb}
\newcommand{\tn}{\textnormal}
\newcommand{\lcm}{\tn{lcm}}

\def\E#1{\ensuremath{\mathbb{E}\left[#1\right]}}
\def\Prob#1{\ensuremath{\mathbb{P}\left(#1\right)}}

\colorlet{ColorYPChen}{blue!80!black}

\newcommand{\poly}{\mathrm{poly}}
\newcommand{\val}{\mathrm{val}}
\newcommand{\polylog}{\mathrm{polylog}}
\newcommand{\rank}{\mathrm{rank}}
\newcommand{\unrank}{\mathrm{unrank}}
\newcommand{\Enc}{\mathrm{Enc}}

\newcommand{\bolda}{\mathbf{a}}
\newcommand{\boldb}{\mathbf{b}}
\newcommand{\boldc}{\mathbf{c}}

\newcommand{\bolde}{\mathbf{e}}

\newcommand{\bolds}{\mathbf{s}}

\newcommand{\boldu}{\mathbf{u}}
\newcommand{\boldv}{\mathbf{v}}

\newcommand{\boldx}{\mathbf{x}}
\newcommand{\boldy}{\mathbf{y}}
\newcommand{\boldz}{\mathbf{z}}

\begin{document}
\title{Correcting Burst Errors with Applications to DNA-Based Data Storage} 

\author{
 \IEEEauthorblockN{Yuan-Pon Chen and Olgica Milenkovic\\
 Coordinated Science Lab \\
 Department of Electrical and Computer Engineering, University of Illinois, Urbana-Champaign}             
}
\maketitle


\begin{abstract} We introduce the problem of permutation burst error-correction, motivated by applications in molecular data storage with nanopore sequencing read platforms. A permutation burst of length $k$ is an error that affects a substring of length $k$ of a string of length $n$ by permuting its entries. Permutation burst errors can be corrected using classical burst error-correcting codes, such as Fire and Reed-Solomon codes. Here, we study how to reduce the redundancy of these classical codes by taking into account that the burst error preserves the weight of the substring. In addition, we further restrict our model to handle cyclic permutation bursts only, which are of independent practical and mathematical interest. Our main results include proving/establishing the following: a) for $k=o(\sqrt{n}),$ the redundancy of any permutation burst error correcting code is at least $\log n-k+\log \binom{2k-2}{k-2}+o(1).$ This result follows from McDiarmid's concentration bound;  b) for $k \geq 9$ and  $k=o\left(\sqrt{\frac{n}{\log n}}\right)$, the redundancy of any cyclic permutation burst error correcting code is at least 
$\log n +\log (k(k-1)-56)-\log\,80+o(1).$ c) the asymptotic redundancy upper bounds for the two types of substring errors that are twice the lower bound; d) constructions for both permutation burst and cyclic permutation burst errors that involve Varshamov-Tenegoltz type constraints or represent generalizations and extensions of Fire and other codes. In particular, we exhibit polynomially encodable and decodable codes that can correct cyclic bursts of length $\leq k$ using  $\log n + k + 49\log k + O(\log \log k)$ redundant bits.
\end{abstract}

\section{Introduction}

DNA-based data storage is an information recording paradigm in which data is encoded using nucleotides of a DNA strand~\cite{church2012next,goldman2013towards,tabatabaei2015rewritable}. One way to retrieve the stored content is through the nanopore readout process~\cite{yazdi2017portable}. A nanopore is a nanometer-wide protein hole that lies across a membrane. By applying a voltage to the membrane and letting the DNA strand pass (translocate) through the pore, the $k$-mer (i.e., $k$-substring) that resides in the pore induces an ionic (electric) current. By reading out the current, one can use specialized signal processing and learning methods~\cite{wick2019performance} to infer the symbol content of the DNA strand. An example pore system used for DNA $k$-mer via immobilization is depicted in Figure~\ref{fig:nanopore}. In this case, the length of the $k-$mer is $k=4$.

Experiments suggest that $k$-mers that are permutations of each other induce similar current signals~\cite{tabatabaei2022expanding}. For instance, as illustrated in Figure~\ref{fig:histogram}, the current levels of the $4$-mers ``2332'', ``2323'', ``2233'', ``3232'' have highly overlapping empirical pdfs (here, the symbols come from an extended $11$-letter DNA alphabet and are for simplicity denoted by numbers). Interestingly, the suffix $3$-mers also represent \emph{cyclic} shifts of each other. These observations motivates the study of error-correcting codes that can handle substring ($k$-mer in the sensing region) permutation and cyclic permutation errors. Two important observations are in place: first, one can resolve potential substring permutation errors based on multiple voltage reads, as was done in~\cite{tabatabaei2022expanding}, or second, by using the fact that one gets a collection of current signals for (almost) every $k-$mer/substring of the string~\cite{kiah2016codes}. The problem with the former approach is that it requires more complicated system designs and longer processing times, while the issue in the latter case is that one does not actually observe the individual symbols in the sensing region but rather a mixture of their responses, which leads to errors in basecallers even when different overlapping substring current signals are observed.

For simplicity, we work over a binary alphabet and delegate the extension to larger alphabets for future work. The error model of interest is defined as follows:
for any input length-$n$ binary sequence, the output is obtained by permuting the bits in a length-$k$ substring of the input sequence. Here, $k$ can either be a constant (which is a more practical choice for real nanopore systems) or grow with $n$ (which is of theoretical interest). We also consider two extensions, pertaining to the permutation errors being cyclic shifts, and multiple substring permutation errors. The cyclic shift model includes both the case where the substring length of the cyclic shift is at most $k$ and the case where the length is exactly $k$.

Formally, let $\mathbf{e}=(e_1,\dots,e_n)\in\mathbb{F}_2^{\,n}$. We say that $\mathbf{e}$ is a (susbstituion) burst of length $k$ if there exists an index $i$ with
$1\le i\le n-k+1$ such that $e_j=0$ for all $j\notin\{i,\dots,i+k-1\}$, $e_i=1$, $e_{i+k-1}=1$, and arbitrary otherwise. A vector is termed an $m$-burst of
length $k$ if it can be written as the sum of at most $m$ bursts, each of
length at most $k$, occurring at pairwise disjoint sets of coordinates. Clearly, a length-$k$ substring permutation error is a special case of a length-$k$ burst substitution error: the added property is that the burst error $\mathbf{e}$ has even weight. It also broadly falls into the class of \emph{substring edit errors}, analyzed in~\cite{li2025optimal}.

A single $k$-burst binary error-correcting code $\mathcal{C}$ of length $n$ is a subset of $\mathbb{F}_2^{\,n}$ such that no $k$-burst can result in codevector confusability. Reiger showed that any such systematic linear code has redundancy at least ~\cite[(8)]{reiger1960codes}\footnote{Note that the analysis in ~\cite{reiger1960codes} accounts for bursts that are allowed to cyclically wrap around the codewords. Our model does not allow for bursts to wrap around, but this is a minor discrepancy for long codelengths $n$.}
\begin{align}
    \log (n-k+2) + k - 1.\label{eq:k_burst_LB}
\end{align}
By adapting the proof of~\cite[(8)]{reiger1960codes} and applying a standard sphere-packing argument,
    it can be shown that the same lower bound in~\eqref{eq:k_burst_LB} also applies to any (not necessarily linear or systematic) $k$-burst error-correcting code. Related results were proved by Abramson~\cite{abramson1959class} and Fire, who also provided code constructions~\cite[Section 8.4]{mceliece2002theory} that can efficiently handle a length-$k$ burst error with redundancy at most
\begin{align}
   \max(\log n + 2k ,3k) + o(1).
   \label{eq:red_Fire_general}
\end{align}
See Appendix~\ref{app:Fire} for a review of Fire codes and the corresponding redundancy analysis.

In general, the modern coding literature refers to
\begin{align}
  |\mathcal{C}| \;\le\; 2^{\,n-2k}\label{eq:Reiger}
\end{align}
as the \emph{Reiger bound}\footnote{For existence results pertaining to codes meeting the Abramson-Reiger bound the interested reader is referred to~\cite{abdel1986existence}. Note that the results pertain to cyclic codes only which limits the length to the form $2^{m}-1,$ where $m\geq 2$ is an integer. Furthermore, the results are nonconstructive except for a hanful of parameter choices.}. The above result asserts that the redundancy has to be at least $\geq 2k,$ and necessarily it has to hold $n\ge 2k$.\footnote{For an error pattern comprising at most $t$ bursts, each of length $k$ or less we similarly have $|\mathcal{C}| \;\le\; 2^{\,n-2tk}$.}
Again, Reiger originally derived
~\eqref{eq:Reiger} for systematic linear codes \cite[(9)]{reiger1960codes},
    but the same bound holds for general length-$k$ burst error-correcting codes as well (see, e.g., \cite[Proof of Theorem 8.10]{mceliece2002theory}.
Combining the ``volume bound'' in \eqref{eq:k_burst_LB} and the Reiger bound in \eqref{eq:Reiger},
    any length-$k$ burst error correcting code has redundancy at least
\begin{align}
    &\max(\log(n-k+2)+k-1,2k)\nonumber\\
    &=\begin{cases}
        \log(n-k+2)+k-1,&\tn{ if }k\leq \log(n-k+2)-1,\\
        2k,&\tn{ if }k> \log(n-k+2)-1.
    \end{cases}
    \label{eq:k_burst_LB_combined}
\end{align}

For multiple bursts, in addition to classical Reed-Solomon code-based approaches, some of the early results in the field include  Stone's codes, based on the Chinese Remainder Theorem~\cite{stone1963multiple} and the bounds by Bridwell and Wolf~\cite{bridwell1970burst}. A more recent work by Villalba and others described properties of multiple burst error correcting MDS codes in~\cite{villalba2016multiple}.

Our results include new lower bounds on the redundancy of even-weight and cyclic permutation bursts, with the former strictly improving the original Reiger-type redundancy bounds for all $k=o(\sqrt{n})$ and the latter improving the even-weight bounds further. These bounds are accompanied by a collection of code constructions for the permutation setting that outperform the lower bound in ~\eqref{eq:k_burst_LB_combined} and the Fire code construction, and generalize some to multiple burst correction settings. 
More precisely, for a cyclic permutation substring error of prime length exactly $k$, one can construct codes whose redundancy is below that of any length-$k$ burst error-correcting code. The ideas behind our proofs range from probabilistic arguments involving McDiarmid's concentration inequality~\cite{mcdiarmid1998concentration}, Varshamov-Tenenholtz (VT) type code equations~\cite{varshamov1965code} and string derivative/integral error-correction. The bounds for the cyclic errors are discussed in Section~\ref{sec:cyclic}

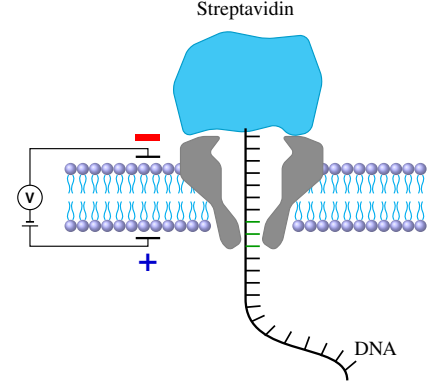
\begin{figure}
    \centering
    \resizebox{.31\textwidth}{!}{
    \begin{tikzpicture}[font=\sffamily]

    
    \foreach \x in {-4.3, -4.0, -3.7, -3.4, -3.1, -2.8, -2.5, -2.2, -1.9, -1.6, -1.3, -1.0} {
        \shade[ball color=blue!40, opacity=0.8] (\x, -1.8) circle (0.15);
        \draw[cyan, thick] (\x-0.05, -1.95) .. controls (\x-0.15, -2.1) and (\x+0.05, -2.2) .. (\x-0.05, -2.4);
        \draw[cyan, thick] (\x+0.05, -1.95) .. controls (\x+0.15, -2.1) and (\x-0.05, -2.2) .. (\x+0.05, -2.4);
        
        \shade[ball color=blue!40, opacity=0.8] (\x, -3.2) circle (0.15);
        \draw[cyan, thick] (\x-0.05, -3.05) .. controls (\x-0.15, -2.9) and (\x+0.05, -2.8) .. (\x-0.05, -2.6);
        \draw[cyan, thick] (\x+0.05, -3.05) .. controls (\x+0.15, -2.9) and (\x-0.05, -2.8) .. (\x+0.05, -2.6);
    }
    
    \foreach \x in {1.3, 1.6, 1.9, 2.2, 2.5, 2.8, 3.1, 3.4, 3.7, 4.0, 4.3} {
        \shade[ball color=blue!40, opacity=0.8] (\x, -1.8) circle (0.15);
        \draw[cyan, thick] (\x-0.05, -1.95) .. controls (\x-0.15, -2.1) and (\x+0.05, -2.2) .. (\x-0.05, -2.4);
        \draw[cyan, thick] (\x+0.05, -1.95) .. controls (\x+0.15, -2.1) and (\x-0.05, -2.2) .. (\x+0.05, -2.4);
        
        \shade[ball color=blue!40, opacity=0.8] (\x, -3.2) circle (0.15);
        \draw[cyan, thick] (\x-0.05, -3.05) .. controls (\x-0.15, -2.9) and (\x+0.05, -2.8) .. (\x-0.05, -2.6);
        \draw[cyan, thick] (\x+0.05, -3.05) .. controls (\x+0.15, -2.9) and (\x-0.05, -2.8) .. (\x+0.05, -2.6);
    }

    \draw[fill=gray!90, draw=gray!100, thick, rounded corners=3pt] 
        (-1.1,-1.0) -- (-1.6,-1.2) -- (-1.6,-2.0) -- (-0.8,-2.7) -- (-0.6,-3.6) -- 
        (-0.4,-3.8) -- (-0.1,-3.6) -- (-0.2,-3.1) -- (-0.6,-2.6) -- (-0.6,-1.7) -- 
        (-0.9,-1.4) -- (-0.6,-1.0) -- cycle;
        
    \draw[fill=gray!90, draw=gray!100, thick, rounded corners=3pt] 
        (1.4,-1.0) -- (1.9,-1.2) -- (1.9,-2.0) -- (1.1,-2.7) -- (0.9,-3.6) -- 
        (0.7,-3.8) -- (0.4,-3.6) -- (0.5,-3.1) -- (0.9,-2.6) -- (0.9,-1.7) -- 
        (1.2,-1.4) -- (0.9,-1.0) -- cycle;

    \draw[fill=cyan!60, draw=cyan!80!black, thick, rounded corners=5pt] 
        (-1.5,-1.0) -- (-1.8, 0.2) -- (-1.4, 1.2) -- (-0.5, 1.6) -- 
        (0.5, 1.4) -- (1.5, 1.5) -- (1.9, 0.4) -- (1.6, -0.9) -- 
        (1.1,-1.0) -- (0.8,-0.7) -- (0,-0.9) -- (-0.8,-0.8) -- cycle;
        
    \node[font=\Large] at (0, 2.1) {Streptavidin};

    \draw[ultra thick] (0, -0.8) -- (0, -5.0) .. controls (0,-6.5) and (2.5,-6.0) .. (2.5, -7.0);
    
    \foreach \y in {-1.0, -1.3, -1.6, -1.9, -2.2, -2.5, -2.8} {
        \draw[very thick] (0,\y) -- (0.35,\y);
    }
    
    \foreach \y in {-3.1, -3.4, -3.7} {
        \draw[very thick, green!60!black] (0,\y) -- (0.35,\y);
    }
    
    \foreach \y in {-4.0, -4.3, -4.6, -4.9} {
        \draw[very thick] (0,\y) -- (0.35,\y);
    }
    
    \draw[very thick] (0.02, -5.2) -- +(5:0.35);
    \draw[very thick] (0.15, -5.55) -- +(25:0.35);
    \draw[very thick] (0.475, -5.9) -- +(45:0.35);
    \draw[very thick] (0.95, -6.1) -- +(60:0.35);
    \draw[very thick] (1.45, -6.25) -- +(70:0.35);
    \draw[very thick] (1.85, -6.4) -- +(65:0.35);
    \draw[very thick] (2.2,  -6.55) -- +(55:0.35);
    \draw[very thick] (2.45, -6.75) -- +(40:0.35);
    
    \node[font=\Large] at (3.2, -6.2) {DNA};

    \draw[thick] (-5.3, -2.5) circle (0.3) node {\textbf{V}};
    
    \draw[thick] (-5.3, -2.2) -- (-5.3, -1.3) -- (-2.4, -1.3) -- (-2.4, -1.5);
    \draw[ultra thick] (-2.7, -1.5) -- (-2.1, -1.5); 
    \fill[red] (-2.7, -0.95) rectangle (-2.1, -1.1);  

    \draw[thick] (-5.3, -2.8) -- (-5.3, -3.1);
    \draw[very thick] (-5.4, -3.1) -- (-5.2, -3.1);
    \draw[thick] (-5.5, -3.2) -- (-5.1, -3.2);
    
    \draw[thick] (-5.3, -3.2) -- (-5.3, -3.7) -- (-2.4, -3.7) -- (-2.4, -3.5);
    \draw[ultra thick] (-2.7, -3.5) -- (-2.1, -3.5); 
    \node[text=blue!80!black, font=\Huge] at (-2.4, -4.1) {\textbf{+}}; 


\end{tikzpicture}
    }
    \caption{A nanopore $k$-mer detection platform, adapted from \cite[Figure 2A]{tabatabaei2022expanding}. A voltage is applied across the membrane and ions cross the two arms of the pore when no DNA translocates through it. When DNA is present, the ion current drops as the molecule ``blocks'' the pore. The $4$-mer within the sensing region/pore constriction is indicated in green.}
    \label{fig:nanopore} \vspace{-0.14in}
\end{figure}

The paper is organized as follows. Section~\ref{sec:general-lower} contains the derivations of lower bounds on the coding redundancy of general single substring permutation error. Cyclic permutation errors are analyzed in Sections~\ref{sec:cyclic}, where lower bounds on the redundancy of the codes are derived using similar arguments as those presented for general permutations. Section~\ref{sec:cyclic-upper} reports upper bounds on the redundancy of both types of codes. Section~\ref{sec:perm-codes} describes simple substring permutation error code constructions for $k=2$ and $k=3$ reminiscent of VT codes. These results are accompanied by straightforward modifications of Fire codes for permutation errors, described in Section~\ref{sec:perm-fire}. Section~\ref{sec:cyclic-refined} presents refined existential results for cyclic substring error-correcting codes, while Section~\ref{sec:cyclic-efficient} presents polynomially encodable/decodable codes. Extensions to multiple general substring permutation errors are briefly described in Section~\ref{sec:multiple-bursts}. 

\begin{figure}
    \centering
    \resizebox{.37\textwidth}{!}{
    \begin{tikzpicture}[font=\sffamily]
    \begin{axis}[
        width=10cm,
        height=6cm,
        xmin=8, xmax=14.5,
        ymin=0, ymax=0.18,
        axis lines=left,
        tick align=outside,
        thick,
        xlabel={$I_{\text{res}}$ (\%)},
        ylabel={Probability},
        every axis y label/.style={
            at={(axis description cs:-0.08, 1.05)},
            anchor=south,
            rotate=0
        },
        xtick={8, 9, 10, 11, 12, 13, 14},
        ytick={0, 0.05, 0.10, 0.15},
        yticklabels={0, 0.05, 0.10, 0.15},
        clip=false 
    ]

    \pgfplotsset{
        peak/.style={
            fill=green!60!black,
            fill opacity=0.6,
            draw=red!90!black,
            draw opacity=1,
            thick,
            samples=60
        }
    }

    
    \addplot[domain=9.5:11.2, peak] {0.142 * exp(-((x-10.2)^2)/0.0072)};
    \addplot[domain=9.5:11.2, peak] {0.148 * exp(-((x-10.25)^2)/0.0072)};
    \addplot[domain=9.5:11.2, peak] {0.140 * exp(-((x-10.35)^2)/0.0072)};
    \addplot[domain=9.5:11.2, peak] {0.130 * exp(-((x-10.45)^2)/0.0072)};

    \addplot[domain=11.5:13.0, peak] {0.125 * exp(-((x-12.2)^2)/0.0072)};
    \addplot[domain=12.5:14.5, peak] {0.115 * exp(-((x-13.6)^2)/0.0072)};

    \node[font=\small] (l1) at (axis cs:9.3, 0.14) {2332};
    \draw (l1) -- (axis cs:10.17, 0.139);
    
    \node[font=\small, anchor=south] (l2) at (axis cs:9.7, 0.153) {2323};
    \draw (l2) -- (axis cs:10.25, 0.147);
    
    \node[font=\small, anchor=south] (l3) at (axis cs:10.6, 0.150) {2233};
    \draw (l3) -- (axis cs:10.35, 0.137);
    
    \node[font=\small] (l4) at (axis cs:11.2, 0.140) {3232};
    \draw (l4) -- (axis cs:10.45, 0.132);

    \node[font=\small, anchor=south] at (axis cs:12.2, 0.125) {3223};
    
    \node[font=\small, anchor=south] at (axis cs:13.6, 0.115) {3322};

    \end{axis}
\end{tikzpicture}
    }
    \caption{Empirical pdfs of ion currents for different orderings of bases within the sensing region. It can be seen that the majority of permutations are indistinguishable. Adapted from \cite[Figure 2C]{tabatabaei2022expanding}.}
    \label{fig:histogram}
\end{figure}
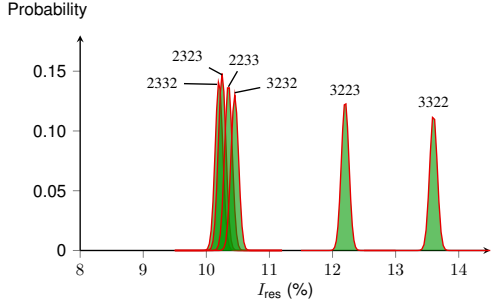

\section{Substring Permutation Codes and Redundancy Lower Bounds} \label{sec:general-lower}

Throughout the text we focus on binary strings of length $n$ and permutation bursts of length $k$. We write $[a,b]=\{{1,1+1,\ldots,b\}}$ and use base-$2$ logarithms unless stated otherwise.

\begin{definition}\label{def:permutation_error}
Let $n,k$ be positive integers, $\mbf{x}=(x_1,\ldots,x_n)\in\{0,1\}^n$ and let $S_k$ stand for the symmetric group of order $k!$.
A sequence $\mbf{y}=(y_1,\ldots,y_n)\in\{0,1\}^n$ is said to be obtained from $\mbf{x}$
    via \emph{a length-$k$ subsequnce permutation error}
    if there exist a permutation $\sigma\in S_k$
    and an index $i\in[1,n-k+1]$ such that
\begin{align*}
    (y_i,\ldots,y_{i+k-1})=(x_{i+\sigma(1)-1},x_{i+\sigma(2)-1},\ldots,x_{i+\sigma(k)-1}),
\end{align*}
    and $y_j=x_j$ for $j\in[1,n]\setminus [i,i+k-1]$.
\end{definition}

\begin{definition}\label{def:code_perm_error}
A code $\mc{C}\subseteq\{0,1\}^n$ is said to be
    an \emph{length-$k$ substring permutation error-correcting code}
    if for any $\mbf{x}\in\mc{C}$ and $\mbf{y}$ obtained from $\mbf{x}$ via a length-$k$ permutation error, we can uniquely recover $\mbf{x}$ from $\mbf{y}$.
\end{definition}

The redundancy of the code $\mc{C}$ equals $n-\log|\mc{C}|$.

Our first results pertains to the redundancy of substring permutation error-correcting codes.
\begin{theorem}\label{thm:perm_error_red_lb}
For $n$ large enough and $k=o(\sqrt{n})$,
    the redundancy of any length-$k$ substring permutation error-correcting code is at least
\begin{align}
    \log n - k + \log \binom{2k-2}{k-2} + o(1).
    \label{eq:thm_k_perm_red_lb}
\end{align}
In particular,
    using the bound $\binom{2k-2}{k-2}\geq 2^{2k-7/2}/\sqrt{k}$ which holds for all $k\geq 2$ (see Appendix~\ref{app:binom_lb}),
    the quantity in~\eqref{eq:thm_k_perm_red_lb} can be lower-bounded by
\begin{align}
    \log n + k - \frac{7}{2}- \frac{1}{2}\log k + o(1).
    \label{eq:thm_k_perm_red_lb_simplified}
\end{align}
\end{theorem}
Note that the difference between the bound corresponding to a substring permutation error and the bound for a general $k$-burst error in \eqref{eq:k_burst_LB_combined} asymptotically reads as
\begin{itemize}
  \item \textbf{$k \leq \log n - 1$}:
  \[
    -\frac{5}{2}-\frac{1}{2}\log k + o(1),
  \]
  \item \textbf{$k \geq \log n - 1$}:
  \[
    -k+\log n -\frac{7}{2}-\frac{1}{2}\log k + o(1).
  \]
\end{itemize}
In both cases, the improvement in the bound is at least $-\frac{5}{2}-\frac{1}{2}\log k$.
It scales as $-\frac{1}{2}\log k$  for the small-$k$ regime, and $-k + \log n$ for the case that $k$ exceeds $\log n$.

\begin{IEEEproof}[Proof of Theorem \ref{thm:perm_error_red_lb}]
We first show that for almost all binary sequences,
    the number of all possible outcomes of a length-$k$ substring permutation error is large.
To this end,
    we first set up some auxiliary definitions
    and then use probabilistic analysis methods.
    
For each $\mbf{x}\in\{0,1\}^n$,
    define $\mc{P}_k(\mbf{x})$
    to be the set of all the sequence that can be obtained from $\mbf{x}$ via a length-$k$ substring permutation error. The number of such sequences is 
    denoted by $\phi_k(\mbf{x})$, and $\phi_k(\mbf{x})=|\mc{P}_k(\mbf{x})|$.
Note that we can compute $\phi_k(\mbf{x})$ as follows. 
For each $i\in[1,n-k+1]$, let
\begin{align*}
    z_{k;i}(\mbf{x})\coloneqq
    \binom{k-1}{w(x_{i+1:i+k-1})+(-1)^{1+x_i}},
\end{align*}
    where $x_{i_1:i_2}\coloneqq (x_{i_1},x_{i_1+1},\ldots,x_{i_2})$ for any two admissible indices $i_1\leq i_2$, and $w(\cdot)$ stands for the Hamming weight.
It is straightforward to see that $z_{k;i}(\mbf{x})$ counts the number of outcomes of length-$k$ substring permutation errors on $\mbf{x},$ where $i$ stands for the index of the first altered bit.
We also use
\begin{align*}
    z_{k;\textnormal{rest}}(\mbf{x})\coloneqq \binom{k-1}{w(x_{n-k+2:n})},
\end{align*}
    which counts the possible permuted outputs obtained from $\mbf{x}$ (including $\mbf{x}$ itself) for indices not covered in the previous case.
It then follows that
\begin{align*}
    \phi_k(\mbf{x})=\sum_{i=1}^{n-k+1}z_{k;i}(\mbf{x})+z_{k;\tn{rest}}.
\end{align*}

We prove next that $\phi_k(\mbf{x})$ is large for almost all $\mbf{x}$ using a probabilistic method.
Let $X=(X_1,\ldots,X_n)$ be uniformly random over $\{0,1\}^n$.
We first calculate the expectation of $\phi_k(X)$,
    denoted by $E$.
By linearity,
    we have
\begin{align}
    E&\coloneqq \E{\phi_k(X)}\nonumber\\
    &=\sum_{i=1}^{n-k+1}\E{z_{k;i}(X)}+\E{z_{k;\tn{rest}}(X)}\nonumber\\
    &=(n-k+1)\E{z_{k;1}(X)}+\E{z_{k;\tn{rest}}(X)},\label{eq:E_phi_k_X}
\end{align}
    where the last equality followed from the fact that all the terms $z_{k;i}(X)$ for $i\in[1,n-k+1]$ have the same distributions.
Then, we note that
\begin{align}
    \E{z_{k;1}(X)}&=\Prob{X_1=1}\E{\binom{k-1}{w(X_{2:k})+1}}\nonumber\\
    &\quad+\Prob{X_1=0}\E{\binom{k-1}{w(X_{2:k})-1}}\nonumber\\
    &=\frac{1}{2}\E{\binom{k-1}{w(X_{2:k})+1}}\nonumber\\
    &\quad+\frac{1}{2}\E{\binom{k-1}{w(X_{2:k})-1}}.\label{eq:E_z_k_1_X}
\end{align}
The first term in \eqref{eq:E_z_k_1_X} can be calculated as follows:
\begin{align}
    \E{\binom{k-1}{w(X_{2:k})+1}}
    &=\sum_{t=0}^{k-1}\Prob{w(X_{2:k})=t}\binom{k-1}{t+1}\nonumber\\
    &=\sum_{t=0}^{k-2}\Prob{w(X_{2:k})=t}\binom{k-1}{t+1}\nonumber\\
    &=\frac{1}{2^{k-1}}\sum_{t=0}^{k-2}\binom{k-1}{t}\binom{k-1}{t+1}\nonumber\\
    &=\frac{1}{2^{k-1}}\sum_{t=0}^{k-2}\binom{k-1}{t}\binom{k-1}{k-t-2}\nonumber\\
    &=\frac{1}{2^{k-1}}\binom{2k-2}{k-2},\label{eq:E_X1is1}
\end{align}
    where the last equality follows from the Vandermonde's identity.
A similar calculation shows that 
\begin{align}
    \E{\binom{k-1}{w(X_{2:k})-1}}=\frac{1}{2^{k-1}}\binom{2k-2}{k-2}.\label{eq:E_X1is0}
\end{align}
Therefore, placing~\eqref{eq:E_X1is1} and~\eqref{eq:E_X1is0} into \eqref{eq:E_z_k_1_X} gives
\begin{align}
    \E{z_{k;1}(X)}=\frac{1}{2^{k-1}}\binom{2k-2}{k-2}.\label{eq:E_z_k_1_X_result}
\end{align}
At the same time,
    we similarly have
\begin{align}
    \E{z_{k;\tn{rest}}(\mbf{x})}
    &=\sum_{t=0}^{k-1}\Prob{w(X_{n-k+2:n})=t}\binom{k-1}{t}\nonumber\\
    &=\frac{1}{2^{k-1}}\binom{k-1}{t}\binom{k-1}{t}\nonumber\\
    &=\frac{1}{2^{k-1}}\binom{2k-2}{k-1}.\label{eq:E_z_k_rest_X_result}
\end{align}
Thus, placing~\eqref{eq:E_z_k_1_X_result} and~\eqref{eq:E_z_k_rest_X_result} into~\eqref{eq:E_phi_k_X} gives
\begin{align}
    E
    =\frac{1}{2^{k-1}}\left((n-k+1)\binom{2k-2}{k-2}+\binom{2k-2}{k-1}\right).
    \label{eq:E_phi_k_X_result}
\end{align}

The next step is to determine how much $\phi_k(\mbf{x})$ changes after flipping one single bit in $\mbf{x}$. More precisely,
    for each $i\in[1,n]$,
    let $\mbf{x}^{(i)}$ be the vector obtained from $\mbf{x}$ by changing the parity of $x_i$. The goal is to find positive constants $c_i$ such that
\begin{align*}
    |\phi_k(\mbf{x})-\phi_k(\mbf{x}^{(i)})|\leq c_i.
\end{align*}
This bound is needed to subsequently invoke a probability concentration bound. 

We first consider the case $i\in[k,n-k+1]$.
Note the changing $x_i$ will only change the values of $z_{k;i-k+1},\ldots,z_{k;i-1}$, and $z_{k;i}$.
For $z_{k;j}$ with $j\in[i-k+1,i-1]$,
    the leading bit $x_j$ is unchanged,
    and the Hamming weight of $x_{j+1:j+k-1}$ changes by one.
Therefore,
    we have
\begin{align*}
    |z_{k;j}(\mbf{x})-z_{k;j}(\mbf{x}^{(i)})|\leq
    d_{k-1},
\end{align*}
    where
\begin{align}
    d_{k-1}= \max_{s\in[0,k-1]}\left|\binom{k-1}{s}-\binom{k-1}{s-1}\right|.
    \label{eq:d_km1_def}
\end{align}
Similarly,
    the effect of changing $z_{k;i}$ can be bounded as
\begin{align*}
    &|z_{k;i}(\mbf{x})-z_{k;i}(\mbf{x}^{(i)})|\\
    &\quad=\left|\binom{k-1}{w(\mbf{x}_{i+1:i+k-1})+1}-\binom{k-1}{w(\mbf{x}_{i+1:i+k-1})-1}\right|\\
    &\quad\leq 2d_{k-1}.
\end{align*}
It therefore follows that
\begin{align}
    |\phi_k(\mbf{x})-\phi_k(\mbf{x}^{(i)})|
    &=\left|\sum_{j=i-k+1}^{i}z_{k;j}(\mbf{x})-z_{k;j}(\mbf{x}^{(i)})\right|\nonumber\\
    &\leq \sum_{j=i-k+1}^{i}|z_{k;j}(\mbf{x})-z_{k;j}(\mbf{x}^{(i)})|\nonumber\\
    &\leq (k+1)d_{k-1}.\label{eq:phi_k_difference_bound}
\end{align}
At the same time,
    changing $x_i$, $i\in[1,k-1],$ can only affect fewer $z_{k;j}$ than those counted in the previous scenario so that the same bound in \eqref{eq:phi_k_difference_bound} applies. Additionally,
    changing $x_i$ for $i\in[n-k+2,n]$ affects at most $k-1$ terms within  $z_{k;j}$, and changing $z_{k;\tn{rest}}$ causes at most a perturbation of $d_{k-1}$. Therefore, for all $i\in[1,n]$, we have
\begin{align}
    c_i=(k+1)d_{k-1}.\label{eq:c_i}
\end{align}

The results of~\eqref{eq:E_phi_k_X_result} and \eqref{eq:c_i} allow us to apply McDiarmid's inequality~\cite{mcdiarmid1998concentration} to get the following claim. 
For any $\varepsilon>0$,
    we have
\begin{align}
    \Prob{\phi_k(X)-E\leq -\varepsilon}\leq \exp\left(-\frac{2\varepsilon^2}{nc_1^2}\right).
    \label{eq:McDiarmid}
\end{align}
In particular, setting 
    $\varepsilon=E/2$
    in~\eqref{eq:McDiarmid}
    gives
\begin{align}
    &\Prob{\phi_k(X)\leq \frac{E}{2}}
    \leq \exp\left(-\frac{E^2}{2nc_1^2}\right).
    \label{eq:McDiarmid_eps_chosen}
\end{align}

Next, we use the classical Hamming-ball argument
    to lower-bound the redundancy.
Let $\mc{C}\subseteq\{0,1\}^n$ be any length-$k$ permutation error-correcting code.
Define
\begin{align*}
    \mc{A}\coloneqq\left\{\mbf{x}\in\{0,1\}^n:\phi_k(\mbf{x})\leq \frac{E}{2}\right\}.
\end{align*}
From \eqref{eq:McDiarmid_eps_chosen} we know that
\begin{align}
    |\mc{A}|\leq 2^n \exp\left(-\frac{E^2}{2nc_1^2}\right).
    \label{eq:A_size_UB}
\end{align}
The cardinality of $\mc{C}$ can be decomposed as
\begin{align}
    |\mc{C}|=|\mc{C}\cap\mc{A}|+|\mc{C}\cap\mc{A}^C|,\label{eq:size_C_decompose}
\end{align}
    where $\mc{A}^C$ denotes the complement of $\mc{A}$.
We already have a bound the first term in \eqref{eq:size_C_decompose} by \eqref{eq:A_size_UB},
    so we next bound the second term in \eqref{eq:size_C_decompose}:
Note that by the unique decodability of $\mc{C}$,
    for each $\mbf{x}\in\mc{C}\cap\mc{A}^C$,
    its ``error ball'' $\mc{P}_k(\mbf{x})$ cannot overlap with $\mc{P}_k(\mbf{y})$ for any other $\mbf{y}\in\mc{C}\cap\mc{A}^C$.
At the same time,
    by the definition of $\mc{A}^C$,
    each of the $\mc{P}_k(\mbf{x})$'s has size at least
    $E/2$.
Therefore,
    we must have
\begin{align*}
    \frac{E}{2}|\mc{C}\cap\mc{A}^C|\leq 2^n,
\end{align*}
    or equivalently,
\begin{align}
    |\mc{C}\cap\mc{A}^C|\leq 2^n\frac{2}{E}.
    \label{eq:C_cap_A_C_UB}
\end{align}
Placing \eqref{eq:A_size_UB} and \eqref{eq:C_cap_A_C_UB} into \eqref{eq:size_C_decompose} yields
\begin{align}
    |\mc{C}|
    &\leq |\mc{A}|+|\mc{C}\cap\mc{A}^C|\nonumber\\
    &\leq 2^n\left(\exp\left(-\frac{E^2}{2nc_1^2}\right)+\frac{2}{E}\right).\label{eq:size_C_UB}
\end{align}

To complete the proof, we establish that for $k=o(\sqrt{n})$,
    $\frac{2}{E}$ is significantly larger than $\exp\left(-\frac{E^2}{2nc_1^2}\right)$.
We split our analysis into two cases:
\begin{enumerate}
    \item
        If $k=O(1)$,
            then $E=\Theta(n)$ and $c_1=O(1)$,
            and the claim clearly holds.
    \item 
        If $k=\omega(1)$ and $k=o(\sqrt{n})$,
            then by Stirling's approximation,
            we have $E=\Theta(n2^k/\sqrt{k})$.
        Furthermore,
            it can be shown that
            $d_{k-1}=\Theta(\frac{2^k}{k})$
            (see Appendix \ref{app:dkm1_asymp}),
            and thus
            $c_1=\Theta(2^k)$. 
        It then follows that
        \begin{align*}
            &-\log\exp\left(-\frac{E^2}{2nc_1^2}\right)=\Theta\left(\frac{n}{k}\right),\\
            &-\log\frac{2}{E}=\log n + k - \frac{1}{2}\log k + O(1)
        \end{align*}
            and thus the claim holds.
\end{enumerate}
In summary,
    for $k=o(\sqrt{n})$, from \eqref{eq:size_C_UB} we have that
\begin{align*}
    |\mathcal{C}|\leq 2^n(\frac{2}{E}(1+o(1)),
\end{align*}
    which implies
\begin{align*}
n-\log|\mathcal{C}|
    &\geq \log E -1 + o(1)\\
    &=\log n - k +\log\binom{2k-2}{k-2} + o(1),
\end{align*}
where we used the fact that $E=\frac{1}{2^{k-1}}(1+o(1))n\binom{2k-2}{k-2}$ from~\eqref{eq:E_phi_k_X_result}.
\end{IEEEproof}

\section{Cyclic Permutation Codes and Redundancy Lower Bounds} \label{sec:cyclic}

We now turn our attention to a special case where the permutation error is induced by a \emph{cyclic} permutation.

Recall that $S_k$ is permutation group over $[1,k]$.
Define $S_{\textnormal{cyc},k}\coloneqq \langle\tau_k\rangle$
    to be the subgroup of $S_k$ generated by $\tau_k$,
    where $\tau_k$ is the permutation defined by
\begin{align*}
    \tau_k: i \mapsto \left(i \bmod k\right) +1.
\end{align*}
\begin{definition}\label{def:cyc_perm_error}
For $\mbf{x}\in\{0,1\}^n$,
    we say that $\mbf{y}$ is obtained from $\mbf{x}$
    via \emph{a length-$k$ cyclic substring permutation error}
    if
    there exists an index $i\in[1,n-k+1]$ and a cyclic permutation $\eta\in S_{\textnormal{cyc},k}$ such that
\begin{align*}
    (y_i,\ldots,y_{i+k-1})=(x_{i+\eta(1)-1},x_{i+\eta(2)-1},\ldots,x_{i+\eta(k)-1}),
\end{align*}
    and where $y_j=x_j$ for $j\in[1,n]\setminus [i,i+k-1]$.
\end{definition}

In what follows, we consider a general substring cyclic permutation error model where the substring has length \emph{at most $k$}. 

\begin{definition}\label{def:leq_k_cyc_perm_error}
We say $\mbf{y}$ is obtained from $\mbf{x}$
    via \emph{a length-$(\leq k)$ cyclic substring permutation error},
    if
    $\mbf{y}$ is obtained from $\mbf{x}$
    via a length-$\ell$ cyclic permutation error
    for some $\ell \leq k$.
\end{definition}

We also restate our definition of an error-correcting code for this particular setting.
\begin{definition}\label{def:code_cyc_perm_error}
A subset $\mc{C}\subseteq\{0,1\}^n$ is said to be
    an \emph{length-$k$ cyclic substring permutation error-correcting code}
    if for any $\mbf{x}\in\mc{C}$ and
    $\mbf{y}$ 
    such that $\mbf{x}\overset{k-\text{cycperm}}{\rightarrow}\mbf{y}$,
    given $\mbf{y}$ we can uniquely recover $\mbf{x}$.
A length-$(\leq k)$ cyclic substring permutation error-correcting code
    is defined in a similar way.
\end{definition}

We have the following lower bound on the redundancy of any length-$(\leq k)$ cyclic substring permutation error-correcting code:
\begin{theorem}\label{thm:CP_red_LB}
For $k\geq 9$ and $k=o(\sqrt{\frac{n}{\log n}})$,
    any length-$(\leq k)$ cyclic substring permutation error-correcting code
    has redundancy at least
\begin{align}
    \log n + \log (k(k-1)-56) - \log 80 + o(1).
    \label{eq:thm_CP_red_LB_exact}
\end{align}
In particular,
    if we further assume $k=\omega(1)$,
    the redundancy lower bound in \eqref{eq:thm_CP_red_LB_exact} can be expressed as
\begin{align}
    \log n + 2\log k + O(1).
    \label{eq:thm_CP_red_LB}
\end{align}
\end{theorem}
We will need the following lemma to prove Theorem \ref{thm:CP_red_LB}.
\begin{lemma}\label{lem:Vz_LB}
Let $\ell\geq 9$ be a positive integer.
For any sequence $\mbf{z}\in\{0,1\}^{\ell}$,
    define $\mc{V}_{\ell}(\mbf{z})$ to be the set of all the sequences $\boldy\in\{0,1\}^{\ell}$ 
    that can be obtained by applying a cyclic permutation in $S_{\tn{cyc},\ell}$ to $\mbf{z}$ such that $\boldz_1\neq \boldy_1$ and $\boldz_{\ell}\neq \boldy_{\ell}$.
If $Z=(Z_1,\ldots,Z_{\ell})$ follows a uniform distribution over $\{0,1\}^{\ell}$, then
\begin{align}
    \E{|\mc{V}_{\ell}(Z)|}\geq \frac{\ell-1}{20}
    \label{eq:lem_Vz_LB}.
\end{align}
\end{lemma}
\begin{IEEEproof}[Proof of Lemma \ref{lem:Vz_LB}]
We first establish a deterministic lower bound on $|\mc{V}_{\ell}(\mbf{z})|$ for any $\mbf{z}\in\{0,1\}^{\ell}$.
For each $s\in[\ell-1]$,
    we write $\tau_{\ell}^s(\mbf{z})$ for the sequence obtained by applying $\tau_{\ell}^s$ to $\mbf{z}$, where if $\mbf{z}=(\mbf{z}_1,\ldots,\mbf{z}_{\ell})$,
    then
\begin{align*}
    \tau_{\ell}^s(\mbf{z})=(\mbf{z}_{1+s},\mbf{z}_{2+s},\ldots,\mbf{z}_{\ell+s}).
\end{align*}
    Here, we identify $\mbf{z}_{i+\ell}$ with $\mbf{z}_i$ for $i\in[1,\ell]$. Then, we define
\begin{align*}
    \mc{Y}_{\ell,s}(\mbf{z})\coloneqq\begin{cases}
        \{\tau_{\ell}^s(\mbf{z})\},\tn{ if }\tau_{\ell}^s(\mbf{z})_1\neq \mbf{z}_1\tn{ and }\tau_{\ell}^s(\mbf{z})_{\ell}\neq \mbf{z}_{\ell}.\\
        \emptyset,\textnormal{ otherwise}.
    \end{cases}
\end{align*}
Consequently,
\begin{align*}
    \mc{V}_{\ell}(\mbf{z})=\bigcup_{s=1}^{\ell-1}\mc{Y}_{\ell,s}(\mbf{z}).
\end{align*}
By the Bonferroni inequality with respect to the counting measure\cite{galambos1977bonferroni},
    we have
\begin{align*}
    |\mc{V}_{\ell}(\mbf{z})|\geq \sum_{s=1}^{\ell-1}|\mc{Y}_{\ell,s}(\mbf{z})| - \sum_{1\leq s_1<s_2\leq \ell-1}|\mc{Y}_{\ell,s_1}(\mbf{z})\cap\mc{Y}_{\ell,s_2}(\mbf{z})|.
\end{align*}
Substituting $\mbf{z}$ with $Z$ and taking expectations on both sides results in
\begin{align}
    \E{|\mc{V}_{\ell}(Z)|}&\geq \sum_{s=1}^{\ell-1}\E{|\mc{Y}_{\ell,s}(Z)|} \nonumber\\
    &~ - \sum_{1\leq s_1<s_2\leq \ell-1}\E{|\mc{Y}_{\ell,s_1(Z)}\cap\mc{Y}_{\ell,s_2}(Z)|}.
    \label{eq:Bonferroni}
\end{align}

We next consider the first sum on the right-hand-side of \eqref{eq:Bonferroni}.
Note that for $s\in[2,\ell-2]$, we simply have 
\begin{align*}
    \E{|\mc{Y}_{\ell,s}(Z)|}
    &=\Prob{\tau_{\ell}^s(Z)_1\neq Z_1\tn{ and }\tau_{\ell}^s(Z)_{\ell}\neq Z_{\ell}}\\
    &=\Prob{Z_{1+s}\neq Z_1\tn{ and } Z_{s}\neq Z_{\ell}}\\
    &=\Prob{Z_{1+s}\neq Z_1}\Prob{Z_{s}\neq Z_{\ell}}\\
    &=\frac{1}{4}.
\end{align*}
We also have
    $\E{|\mc{Y}_{\ell,1}(Z)|}=\Prob{Z_2\neq Z_1 \tn{ and }Z_1\neq Z_{\ell}}$,
    which is still $\frac{1}{4}$.
Similarly, $\E{|\mc{Y}_{\ell,\ell}(Z)|}=\frac{1}{4}$
    as well.
It then follows that 
$\sum_{s=1}^{\ell-1}\E{|\mc{Y}_{\ell,s}|}=\frac{\ell-1}{4}$.

We now upper-bound the second sum on the right-hand-side of \eqref{eq:Bonferroni}.
Note that for each $1\leq s_1<s_2\leq \ell-1$,
    we have
\begin{align*}
    \E{|\mc{Y}_{\ell,s_1}(Z)\cap\mc{Y}_{\ell,s_2}(Z)|}
    &\leq\E{|\{\tau_{\ell}^{s_1}(Z)\}\cap\{\tau_{\ell}^{s_2}(Z)\}|}\\
    &=\Prob{\tau_{\ell}^{s_1}(Z)=\tau_{\ell}^{s_2}(Z)}.
\end{align*}
Also, note that
    if a sequence $\mbf{z}\in\{0,1\}^{\ell}$ satisfies
    $\tau_{\ell}^{s_1}(\mbf{z})=\tau_{\ell}^{s_2}(\mbf{z})$,
    then $\mbf{z}$ has (not necessarily the least) period $\tn{gcd}(\ell,s_2-s_1)$. It can be observed that for any $p|\ell$
    the probability that $Z$ has period $p$
    is $2^{-\ell+p}$,
    since there are only $p$ ``free bits''.
It therefore follows that
\begin{align*}
    \E{|\mc{Y}_{\ell,s_1}(Z)\cap\mc{Y}_{\ell,s_2}(Z)|}
    &\leq 2^{-\ell+\tn{gcd}(\ell,s_2-s_1)}\\
    &\leq 2^{-\frac{\ell}{2}},
\end{align*}
    where we used the fact that $\tn{gcd}(\ell,s_2-s_1)\leq \ell/2$ since $s_2-s_1<\ell$.
It follows that the second sum in \eqref{eq:Bonferroni}
    is upper-bounded by
    $\frac{\ell(\ell-1)}{2}2^{-\frac{\ell}{2}}$.

Putting these two results into
    \eqref{eq:Bonferroni} leads to
\begin{align*}
    \E{|\mc{V}_{\ell}(Z)|}& \geq \frac{\ell-1}{4}-\frac{\ell(\ell-1)}{2}2^{-\frac{\ell}{2}}\\
    &=\frac{\ell-1}{4}(1-\ell2^{-\frac{\ell}{2}+1}).
\end{align*}
By differentiating, we get $1-\ell2^{-\frac{\ell}{2}+1}\geq \frac{1}{5}$ for $\ell\geq 9$,
    which establishes~\eqref{eq:lem_Vz_LB}.
\end{IEEEproof}

\begin{IEEEproof}[Proof of Theorem~\ref{thm:CP_red_LB}]
Similarly to the proof of Theorem~\ref{thm:perm_error_red_lb},
    we show that almost every input sequence $\mbf{x}\in\{0,1\}^n$ has an error-guard ball of size $\Omega(nk^2)$, so that the desired result follows from the standard sphere-packing arguments. The idea is to find certain disjoint subsets of the error balls which facilitate using the linearity of expectation to lower-bound the size of each subset within  Lemma~\ref{lem:Vz_LB}.
    
Before we introduce any randomness in the analysis, we first set up some deterministic notions and rules.
For each $\mbf{x}\in\{0,1\}^n$,
    we define
    $\mc{CP}_{\leq k}(\mbf{x})$ to be the set of all the sequences that can be obtained from $\mbf{x}$
    via a length-$(\leq k)$ cyclic permutation, while excluding $\mbf{x}$ itself. That is,
\begin{align*}
    \mc{CP}_{\leq k}(\mbf{x}) = \bigcup_{\ell=2}^k \bigcup_{i=1}^{n-\ell+1} \mc{CP}_{\ell;i}(\mbf{x}),
\end{align*}
    where $\mc{CP}_{\ell;i}(\mbf{x})$ is the set of all sequences obtained from $\mbf{x}$ via applying a cyclic permutation to the window $(x_{i},\ldots,x_{i+\ell-1})$ that results in a string different from $\mbf{x}$ itself. From each $\mc{CP}_{\ell;i}(\mbf{x})$ we collect strings such that
    the two boundary points of the window are altered (flipped). More precisely, define
\begin{align*}
    \mc{U}_{\ell;i}(\mbf{x})\coloneqq\{\mbf{y}\in\mc{CP}_{\ell;i}(\mbf{x})~:~y_i\neq x_i, y_{i+\ell-1}\neq x_{i+\ell-1}\}.
\end{align*}
It follows that
    any two distinct $\mc{U}_{\ell;i}$ and $\mc{U}_{\ell';j}$ are disjoint,
    and 
\begin{align}
    |\mc{CP}_{\leq k}(\mbf{x})|
    &\geq \left|\bigcup_{\ell=2}^k \bigcup_{i=1}^{n-\ell+1} \mc{U}_{\ell;i}(\mbf{x})\right|\label{eq:CPleqk_LB_ineq}\\
    &=\sum_{\ell=2}^k \sum_{i=1}^{n-\ell+1} |\mc{U}_{\ell;i}(\mbf{x})|.\label{eq:CPleqk_LB}
\end{align}
For numerical examples of $\mc{CP}_{\ell;i}$ and $\mc{U}_{\ell;i}$, see Example~\ref{ex:CP_U}.

\begin{remark}
Note that the inequality in \eqref{eq:CPleqk_LB_ineq} is strict: consider $\mbf{x}=110011101$.
Then $\mbf{y}=101110011\in\mc{CP}_{9;1}(\mbf{x})$ but $\mbf{y}$ is not in any $\mc{U}_{\ell;i}(\mbf{x})$ (by locating the first and the last flipped bit in $\mbf{y}$, the only possible candidate is $\mc{U}_{7;2}(\mbf{x})$, and it can be checked that $\mbf{y}\notin \mc{U}_{7;2}(\mbf{x})$).
\end{remark}

Now consider $X$ to be uniformly random over $\{0,1\}^n$.
We seek to find a lower bound on $\E{|\mc{CP}_{\leq k}(X)|}$.
Taking expectation on both sides of \eqref{eq:CPleqk_LB} yields
\begin{align*}
    \E{|\mc{CP}_{\leq k}(X)|} \geq \sum_{\ell=2}^k \sum_{i=1}^{n-\ell+1} \E{|\mc{U}_{\ell;i}(X)|}.
\end{align*}
Note that each $|\mc{U}_{\ell;i}(X)|$ follows the same distribution as $|\mc{V}_{\ell}(Z)|$ defined in Lemma \ref{lem:Vz_LB}.
Therefore, Lemma~\ref{lem:Vz_LB} yields
\begin{align}
    \E{|\mc{CP}_{\leq k}(X)|}
    &\geq \sum_{\ell=9}^k (n-\ell+1)\frac{\ell-1}{20}\nonumber\\
    &=(1+o(1))\frac{k(k-1)-56}{40}n\label{eq:ECPleqkX_LB_exact}\\
    &=\Omega(nk^2).\label{eq:ECPleqkX_LB}
\end{align}

We then proceed to apply McDiarmid's inequality. The first step is to  determine the Lipschizt constant of $|\mc{CP}_{\leq k}(\mbf{x})|$.
For each $\ell\in[2,k]$,
    any bit of $\mbf{x}$ is in
    at most $\ell$ different windows of length $\ell$,
    and by flipping that bit the number of possible outputs in each window changes by at most $\ell$.
It follows that the Lipschitz constant of $|\mc{CP}_{\leq k}(\mbf{x})|$ can be taken to be
\begin{align*}
    \sum_{\ell=2}^k \ell^2=\Theta(k^3).
\end{align*}
Therefore,
    if we denote $E\coloneqq\E{|\mc{CP}_{\leq k}(X)|}$ and define 
    $\mc{A}\coloneqq\left\{\mbf{x}\in\{0,1\}^n~:~|\mc{CP}_{\leq k}(\mbf{x})|\leq \frac{E}{2}\right\}$,
    McDiarmid's inequality leads to: 
\begin{align*}
    \Prob{X\in\mc{A}}
    &\leq \exp\left(-\frac{E^2}{2n\Theta(k^3)^2}\right)\\
    &\leq \exp\left(-\frac{\Omega(nk^2)^2}{\Theta(nk^6)}\right)\\
    &=\exp\left(-\Omega(\frac{n}{k^2})\right).
\end{align*}
It subsequently follows that $|\mc{A}|\leq 2^n\exp\left(-\Omega(\frac{n}{k^2})\right)$.

As the last step,
    let $\mc{C}$ be a length-$(\leq k)$ cyclic substring permutation error-correcting code.
By unique decodability of $\mc{C}$,
    we have
\begin{align}
    |\mc{C}|
    &=|\mc{C}\cap\mc{A}|+|\mc{C}\cap\mc{A}^C|\nonumber\\
    &\leq 2^n\left(\exp\left(-\Omega(\frac{n}{k^2})\right)+\frac{2}{E}\right)\nonumber\\
    &=\frac{2^{n+1}}{E}(1+o(1)),\label{eq:C_size_UB}
\end{align}
    where we used the fact that $E=\Omega(nk^2)$ from \eqref{eq:ECPleqkX_LB} and the assumption that $k=o(\sqrt{\frac{n}{\log n}})$.
Finally,
    from \eqref{eq:C_size_UB} and \eqref{eq:ECPleqkX_LB_exact} we get \eqref{eq:thm_CP_red_LB_exact}.
\end{IEEEproof}
\begin{example}\label{ex:CP_U}
The following examples pertain to the sets $\mc{CP}_{\ell,i}$ and $\mc{U}_{\ell,i}$
    defined in the proof of Theorem~\ref{thm:CP_red_LB}.
    
Let $k=4$ and $\mbf{x}=10010$.
Then
\begin{align*}
    &\mc{CP}_{2,1}=\{01010\}, ~ \mc{CP}_{2,2}=\emptyset, ~ \mc{CP}_{2,3}=\{10100\}, \\ 
    &\mc{CP}_{2,4}=\{10001\}, ~ \mc{CP}_{3,1}=\{01010,00110\}, \\
    &\mc{CP}_{3,2}=\{11000,10100\},
    ~\mc{CP}_{3,3}=\{10100,10001\}, \\
    &\mc{CP}_{4,1}=\{11000,01100,00110\}, \\
    &\mc{CP}_{4,2}=\{11000,10100,10001\}.
\end{align*}
Meanwhile,
\begin{align*}
    &\mc{U}_{2,1}=\{01010\}, ~ \mc{U}_{2,2}=\emptyset, ~ \mc{U}_{2,3}=\{10100\}, \\ 
    &\mc{U}_{2,4}=\{10001\}, ~ \mc{U}_{3,1}=\{00110\}, \\
    &\mc{U}_{3,2}=\{11000\},
    ~\mc{U}_{3,3}=\emptyset, \\
    &\mc{U}_{4,1}=\{01100\}, 
    ~\mc{U}_{4,2}=\emptyset.
\end{align*}
It can be seen that each $\mc{U}_{\ell,i}$ is a subset of $\mc{CP}_{\ell,i}$ and that the $\mc{U}_{\ell,i}$s are pairwise disjoint.
\end{example}

Following almost the same line of proof,
    we can also get the necessary redundancy lower bound
    when the substring length is \emph{exactly} $k$.
\begin{theorem}\label{thm:CP_exact_k_red_LB}
For $k\geq 9$ and $k=o(\sqrt{\frac{n}{\log n}})$,
    any length-$k$ cyclic permutation error-correcting code
    has redundancy at least
\begin{align}
    \log n + \log (k-1) - \log 40 + o(1).
    \label{eq:thm_CP_exact_k_red_LB_exact}
\end{align}
In particular,
    if we further assume $k=\omega(1)$,
    the redundancy lower bound in \eqref{eq:thm_CP_exact_k_red_LB_exact} can be expressed as
\begin{align}
    \log n + \log k + O(1).
    \label{eq:thm_CP_exact_k_red_LB}
\end{align}
\end{theorem}
\begin{IEEEproof}[Sketch of the proof]
Let $\mc{CP}_k(\mbf{x})$ be the set of all the sequences that can be obtained from $\mbf{x}$ via a length-$k$ cyclic permutation, excluding $\mbf{x}$ itself.
By counting those patterns where the endpoints of the substring are flipped,
    for $X$ following the uniform distribution over $\{0,1\}^n$,
    and invoking Lemma~\ref{lem:Vz_LB} leads to
\begin{align*}
    \E{|\mc{CP}_k(X)|}
    &\geq \frac{(n-k+1)(k-1)}{20}\\
    &=\Omega(nk).
\end{align*}

The Lipschitz constant of $|\mc{CP}_k(\mbf{x})|$ is $k^2$,
    and by McDiarmid's inequality,
    the probability that
    $\E{|\mc{CP}_k(X)|}$ is below half its mean is still $\exp\left(-\Omega(\frac{n}{k^2})\right)$.
A sphere-packing argument then proves~\eqref{eq:thm_CP_exact_k_red_LB_exact}.
\end{IEEEproof}

\section{Redundancy Upper Bounds for Permutation Codes} \label{sec:cyclic-upper}

The reason behind the somewhat technical derivations needed to establish  lower bounds is that even though it is easy to see that the worst-case error ball is of size $\Theta(nk^2)$ (for the case of substring lengths at most $k$), one still needs to show that the average size is $\Theta(nk^2)$.
However, to establish upper bounds, we only need to upper-bound the size of the error balls, and thus the arguments become much simpler.

As one may expect, we have the following redundancy upper bound that is asymptotically twice the lower bound.
\begin{theorem}\label{thm:CP_red_UB} Assume $2\leq k \leq n$. Then, there exists a 
    length-$(\leq k)$ cyclic permutation error-correcting code 
    with redundancy at most
\begin{align}
    2\log(n-\frac{2k-1}{3})+2\log(k(k-1))-1.
    \label{eq:thm_CP_red_UB_exact}
\end{align}
In particular,
    if we assume that $k=\omega(1)$ and $k=o(n)$ hold\footnote{We need $k$ to be large enough to ensure $\log (k(k-1))=2\log k + o(1)$,
        and we need $k=o(n)$ to claim $\log(n-\frac{2k-1}{3})=\log n+o(1)$.}
    then the redundancy in \eqref{eq:thm_CP_red_UB_exact} can be expressed as 
\begin{align}
    2\log n +4\log k -1 + o(1).
    \label{eq:thm_CP_red_UB}
\end{align}
\end{theorem}
\begin{IEEEproof}
For each $\mbf{x}\in\{0,1\}^n$,
    define
    $\mc{B}_{\leq k}(\mbf{x})$
    to be the set of all sequences
    that can result in a sequence in $\mc{CP}_{\leq k}(\mbf{x})$
    via a length-$(\leq k)$ cyclic substring permutation error (including $\boldx$ itself). 
The size of $\mc{B}_{\leq k}(\mbf{x})$ can be easily upper-bounded as follows:
for each $\ell\in[2,k],$ there are $n-\ell+1$ substrings of length $\ell$ in $\mbf{x}$,
    and each such sunstring can result in at most $\ell-1$ true error patterns.
Therefore,
    the size of $\mc{CP}_{\leq k}(\mbf{x})$ is upper-bounded by
\begin{align*}
    \sum_{\ell=2}^k (n-\ell+1)(\ell-1)=\frac{k(k-1)}{2}(n-\frac{2k-1}{3}).
\end{align*}
By the symmetry of our model,
    we have $\mbf{y}\in\mc{CP}_{\leq k}(\mbf{x})$ if and only if $\mbf{x}\in\mc{CP}_{\leq k}(\mbf{y})$.
Therefore,
    the size of $\mc{B}_{\leq k}(\mbf{x})\setminus\{\boldx\}$ is upper-bounded by
\begin{align}
    \left(\frac{k(k-1)}{2}(n-\frac{2k-1}{3})\right)^2.
    \label{eq:Bleqk_x_exclude_x_ub}
\end{align}
To establish this claim, let $F\coloneqq \frac{k(k-1)}{2}(n-\frac{2k-1}{3})$.
For each $\mbf{y}\in\mc{CP}_{\leq k}(\mbf{x})$,
    the number of sequence in $\mc{CP}_{\leq k}(\boldy)\setminus\{\boldx\}$ is at most $F-1$,
    and by including $\boldy$ itself the number of sequence in $(\mc{CP}_{\leq k}(\boldy)\cup\{\boldy\})\setminus\{\boldx\}$ is at most $F$.
Since there are at most $F$ choices of such $\boldy$,
    the upper bound in~\eqref{eq:Bleqk_x_exclude_x_ub} holds. It therefore follows that
\begin{align*}
    \left|\mc{B}_{\leq k}(\mbf{x})\right|
    &\leq 1+\left(\frac{k(k-1)}{2}(n-\frac{2k-1}{3})\right)^2\\
    &\leq 2\left(\frac{k(k-1)}{2}(n-\frac{2k-1}{3})\right)^2,
\end{align*}
    where the last inequality is a consequence of the fact that $\frac{k(k-1)}{2}(n-\frac{2k-1}{3})=\sum_{\ell=2}^k(n-\ell+1)(\ell-1)\geq 1$ for $2\leq k\leq n$. A greedy algorithm that picks sequences from $\{0,1\}^n$ while eliminating all other sequences in $\mc{B}_{\leq k}(\mbf{x})$ left after selecting $\mbf{x}$
    leads to the desired code.
\end{IEEEproof}

Similarly,
    we present an analogous result for the case of substring errors of length exactly equal to $k$, but omit the proof.
\begin{theorem}\label{thm:CP_exact_k_red_UB}
Assume $2\leq k \leq n$. Then there exists a 
    length-$k$ cyclic substring permutation error-correcting code 
    with redundancy at most
\begin{align}
    2\log(n-k+1)+2\log(k-1)+1.
    \label{eq:thm_CP_exact_k_red_UB_exact}
\end{align}
In particular,
    if we assume that $k=\omega(1)$ and $k=o(n)$, then the redundancy in~\eqref{eq:thm_CP_exact_k_red_UB_exact} can be expressed as 
\begin{align}
    2\log n +2\log k + o(1).
    \label{eq:thm_CP_exact_k_red_UB}
\end{align}
\end{theorem}

Combining
    Theorem \ref{thm:CP_red_UB}
    and Equation \eqref{eq:k_burst_LB},
    we have the following result.
\begin{corollary}
If $k=o(\sqrt{\frac{n}{\log n}})$ and $k>(1+\varepsilon)\log n$
    for some $\varepsilon>0$,
    then there exists a length-$(\leq k)$ cyclic substring permutation error-correcting code whose redundancy is strict smaller than that of  \emph{any} length-$k$ burst error-correcting code.
\end{corollary}

The counterpart result
    on the redundancy upper bound of
    length-$k$ permutation error-correcting code reads as follows.
\begin{proposition}
Assume $2\leq k \leq n$. Then there exists a length-$k$  permutation error-correcting code
    with redundancy at most
\begin{align}
    2\log(n-k+1)+2\log\binom{k}{k/2} +1.
    \label{eq:prop_perm_k_UB}
\end{align}
In particular,
    if we assume $k=\omega(1)$ and $k=o(n)$,
    then Equation \eqref{eq:prop_perm_k_UB} simplifies to
\begin{align*}
    2\log n + 2k - \log \frac{\pi k}{2} +1 + o(1).
\end{align*}
\end{proposition}
\begin{proof}
Similar to the proof of Theorem \ref{thm:CP_red_UB},
    for each $\boldx\in\{0,1\}^n$,
    it can be shown that
    there are at most $(n-k+1)\binom{k}{k/2}$ sequences 
    not equal to $\boldx$
    that can be obtained from $\boldx$ via a length-$k$ permutation error.
To see this,
    there are $n-k+1$ ways to select a length-$k$ substring in $\boldx$,
    and each substring has at most $\binom{k}{k/2}$ different permuted outcomes. Then,
    by the symmetric of the error model,
    the number of sequences $\boldz\in\{0,1\}^n$
    that can have the same output as $\boldx$
    (including $\boldx$ itself)
    is upper-bounded by $1+(n-k+1)^2\binom{k}{k/2}^2\leq 2(n-k+1)^2\binom{k}{k/2}^2$. A greedy algorithm establishes a code with the desired redundancy upper bound in~\eqref{eq:prop_perm_k_UB}.
\end{proof}

\section{Permutation-Correcting Codes for $k=2$ and $3$} \label{sec:perm-codes}

The problem of constructing codes for a single substring permutation error of redundancy close to the one derived in the previous section is challenging. Straightforward approaches that use the fact that permutations preserve the weight of the erroneous substring in conjunction with Fire codes and other burst error-correcting techniques do not offer significant redundancy savings, as will be seen in the next section.

This is why we start with constructions suitable for small constant values of $k$ (e.g., $k=2$ and $3$), which rely on using VT-like deletion-correcting codes. At the end of our discussion we explain why such code constructions may require significantly more thought so as to extend to the case $k \geq 4$.

For $\mbf{x}=(x_1,\ldots,x_n)\in\{0,1\}^n$,
    define the hash functions
\begin{align*}
    f_1(\mbf{x})&\coloneqq\sum_{i=1}^n ix_i,\\
    f_2(\mbf{x})&\coloneqq\sum_{i=1}^n\binom{i}{2}x_i.
\end{align*}
We start with $k=2$. In this case, the permutation error reduces to an adjacent transposition of bits, and codes correcting such errors in the context of data storage and runlength coding were first described in~\cite{levenshtein1993perfect} under the name of \emph{peak shift errors}. More recent work~\cite{gabrys2017codes} extends this approach to adjacent blocks of transpositions and deletion errors, while the simple case of a single transposition is not handled in a 
separate context. We have the following result.
\begin{proposition}
Knowing $f_2(\mbf{x})\mod 2n-1$ allows one to uniquely recover $\mbf{x}$ up to a permutation error of length $k=2$.
This leads to a substring permutation-error-correcting code of redundancy
\begin{align*}
    \log n + 1,
\end{align*}
where we excluded the $o(1)$ term, which is smaller than the Fire-code lower bound $\log n + 2k - 1 - \log(2k-1)\approx \log n + 1.415$ for large $n$ from~\eqref{eq:Fire_lb} in Appendix \ref{app:Fire}.
\end{proposition}
The proof of this result is obvious given that an adjacent transposition of substrings containing one $0$ and one $1$ changes the sum by $i$, thereby indicating the exact position of the switch. Similarly, we have the following result for $k=3$.
\begin{proposition}\label{prop:code_3_Perm}
Knowing $f_2(\mbf{x})\mod 2n-1$ and $f_1(\mbf{x})\mod 2$ allows one to uniquely recover $\mbf{x}$ up to a $k=3$-substring permutation error.
This leads to a permutation-error-correcting code of redundancy
\begin{align*}
    \log n + 2,
\end{align*}
where we excluded the $o(1)$ term, which is smaller than the 
Fire-code lower bound $\log n + 2k - 1 - \log(2k-1)\approx \log n + 2.678$ for large $n$ from \eqref{eq:Fire_lb} in Appendix~\ref{app:Fire}.
\end{proposition}
\begin{IEEEproof}
The key observation is that 
there are actually only $4$ ``distinct'' error patterns one has to consider for correction.
First, for any length-$3$ binary string $\mbf{x}$, there are at most two other length-$3$ strings that have the same Hamming weight as $\mbf{s}$: the error patterns are either adjacent transpositions or distance-$2$ transpositions (reversals), the latter of the form
\begin{align}
\begin{split}
    0*1 &\mapsto 1*0\\
    1*0 &\mapsto 0*1,
\end{split}\label{eq:patterns_3_Perm_non_adj_trans}
\end{align}
    where $*$ stands for either $0$ or $1$.

\begin{table}[]
    \centering
    \begin{tabular}{c||c|c|c}
        Pattern & $\mathcal{D} f_2$ & $\mathcal{D} f_1$ & Range of the starting index $i$\\
        \hline
        $10\mapsto 01$ &  $i$ &  $1$ & $[1,n-1]$\\
        $01\mapsto 10$ & $-i$ & $-1$ & $[1,n-1]$\\
        $1*0\mapsto 0*1$ & $2i+1$ & $2$ & $[1,n-2]$\\
        $0*1\mapsto 1*0$ & $-2i-1$ & $-2$ & $[1,n-2]$
    \end{tabular}
    \caption{The change of the hash functions $f_2$ and $f_1$ from the original sequence to the altered one.
    For each pattern,
        we assume the starting bit is the $i$th bit of the input sequence $\mbf{x}$.}
    \label{tab:Df_3_Perm}
\end{table}
\begin{table}[]
    \centering
    \begin{tabular}{c||c|c|c}
        Pattern & $\mathcal{D} f_2$ & $\mathcal{D} f_1$ \\
        \hline
        $10\mapsto 01$ &  $\{1,2,\ldots,n-1\}$ &  $1$ \\
        $01\mapsto 10$ & $\{-1,-2,\ldots,-n+1\}$ & $-1$ \\
        $1*0\mapsto 0*1$ & $\{3,5,\ldots,2n-3\}$ & $2$ \\
        $0*1\mapsto 1*0$ & $\{-3,-5,\ldots,-2n+3\}$ & $-2$ 
    \end{tabular}
    \caption{All the possible changes of the hash functions $f_2$ and $f_1$ from the original sequence to the altered one.}
    \label{tab:Df_3_Perm_all}
\end{table}
\begin{table}[]
    \centering
    \begin{tabular}{c||c|c|c}
        Pattern & $\mathcal{D} f_2\bmod (2n-1)$ & $\mathcal{D} f_1\bmod 2$ \\
        \hline
        $10\mapsto 01$ &  $\{1,2,\ldots,n-1\}$ &  $1$ \\
        $01\mapsto 10$ & $\{2n-2,2n-3,\ldots,n\}$ & $1$ \\
        $1*0\mapsto 0*1$ & $\{3,5,\ldots,2n-3\}$ & $0$ \\
        $0*1\mapsto 1*0$ & $\{2n-4,2n-6,\ldots,2\}$ & $0$ 
    \end{tabular}
    \caption{All the possible changes of the hash functions $f_2$ and $f_1$ after modulo $2n-1$ and $2$ reduction, respectively.}
    \label{tab:Df_3_Perm_all_mod}
\end{table}
It follows that
    we only have to correct adjacent transpositions and the two reversals in \eqref{eq:patterns_3_Perm_non_adj_trans}.
To this end,
    we record the change of the hash functions $f_2$ and $f_1$ from the original sequence to the altered one (denoted by $\mathcal{D} f_2$ and $\mathcal{D} f_1$, respectively).
The results are listed in Table~\ref{tab:Df_3_Perm}. Using Table ~\ref{tab:Df_3_Perm}, we then list 
    all the possible differences of the two hash functions in Table~\ref{tab:Df_3_Perm_all}, with the differences modulo $2n-1$ and $2$ respectively summarized in Table~\ref{tab:Df_3_Perm_all_mod}.

Then, the result that unique recovery of $\mbf{x}$ is possible based on $f_2(\mbf{x})\bmod 2n-1$ and $f_1(\mbf{x})\bmod 2$ after a length-$3$ cyclic permutation is evident from Table \ref{tab:Df_3_Perm_all_mod}:
we first calculate $d_2\coloneqq\mathcal{D} f_2\bmod (2n-1)$ and the consider the following two cases:
\begin{enumerate}
    \item If $d_2=0$, then we declare that no error has happened.
    \item If $d_2\neq 0$,
        then we calculate $d_1\coloneqq \mathcal{D} f_1\bmod 2$ and consider the following two sub-cases:
        \begin{enumerate}
            \item If $d_1=1$,
                then we know that there an adjacent transposition error has happened.
                The starting location of that error is
                \begin{align*}
                \begin{cases}
                    d_2,&\textnormal{ if }d_2\leq n-1,\\
                    2n-1-d_2,&\textnormal{ if }d_2\geq n.
                \end{cases}
                \end{align*}
                We perform an adjacent transposition to recover $\mbf{x}$.
            \item If $d_1=0$,
                then we know that a reversal error has happened. The position of the starting bit is
                \begin{align*}
                \begin{cases}
                    \frac{d_2-1}{2},&\textnormal{ if }d_2\textnormal{ is odd},\\
                    n-1-\frac{d_2}{2},&\textnormal{ if }d_2\textnormal{ is even}.
                \end{cases}
                \end{align*}
                We then reverse the string to recover $\mbf{x}.$
        \end{enumerate}
\end{enumerate}
These arguments conclude the proof of Proposition \ref{prop:code_3_Perm}.
\end{IEEEproof}
\begin{remark}
Note that all the four patterns in Table \ref{tab:Df_3_Perm} can be obtained via a length-$3$ cyclic permutation.
Therefore,
    a code can correct a length-$3$ permutation if and only if 
    it can correct a length-$3$ cyclic permutation.
\end{remark}
However, for $k=4$ or larger, we have to include additional patterns, such as the substring complement (e.g., $0110\rightarrow 1001$). This kind of error can only be detected by $f_1$ and $f_2$ but not corrected.
To be more precise,
    let $\mbf{y}$ be obtained from $\mbf{x}$ by replacing a substring $0110$ in $\mbf{x}$ with $1001$.
It can be shown that
\begin{align*}
    &f_1(\mbf{y})-f_1(\mbf{x})=0,\\
    &f_2(\mbf{y})-f_2(\mbf{x})=2,
\end{align*}
so that we can only learn that some $0110$ had been replaced by $1001$ (i.e., complemented) but would not know where the error happened. As we are dealing with permutation errors, the only allowed complementation error is one that affects substrings of weight two.

Correcting the above complementation errors for k=4 can be accomplished in a straightforward manner, as illustrated below, using codewords derivatives and shortened Hamming codes. Nevertheless, the increase in redundancy needed to fix this very specific complementation error is $\log\,n$, which exceeds the redundancy of Fire codes for sufficiently large values of $n$. Hence, one needs to be able to jointly recover from complementation and other types of permutation errors to avoid increasing the redundancy beyond the bound dictated by Fire codes. 

The key observation needed to fix this problem is that complementing a substring acts very
simply on the \emph{discrete derivative} (transition sequence) of the codeword. Since we will make use of the discrete derivative in latter derivations as well, a formal definition is provided next.

\begin{definition}\label{def:diff}
For $\bc \in \{0,1\}^n$ define $\Delta(\bc) \in \{0,1\}^{n+1}$ as
\[
  \Delta(\bc)_j\coloneqq c_{j-1} \xor c_j, \; j = 1,\dots,n+1,
  \; \text{where } c_0 = c_{n+1} = 0 .
\]
\end{definition}

\begin{proposition}
\label{lem:bij}
The map $\Delta$ is a bijection from $\{0,1\}^n$ onto the set of even-weight
vectors in $\{0,1\}^{n+1}$.
\end{proposition}

\begin{proof}
Clearly, $\bigoplus_{j=1}^{n+1} d_j = c_0 \xor c_{n+1} = 0$, so every image has even
weight. Conversely $\bc$ can be recovered from $\Delta$ by the prefix sum
$c_j = \bigoplus_{r=1}^{j} \Delta_r$, which inverts $\Delta$. Both spaces have size
$2^n$, so the (injective) map is a bijection onto the even-weight subspace.
\end{proof}

\begin{lemma}
\label{lem:deriv}
For any $\bc \in \{0,1\}^n$ and any $i \in \{1,\dots,n-k+1\}$,
\[
  \Delta\!\big(\mathrm{c}_{k,i}(\bc)\big)
  \;=\; \Delta(\bc) \xor \big(\boldsymbol{\delta}_i \xor \boldsymbol{\delta}_{i+k}\big),
\]
where $\boldsymbol{\delta}_p \in \{0,1\}^{n+1}$ is the $p$-th unit vector and $\mathrm{c}_{k,i}(\bc)$ corresponds to the codestring $\bc$ corrupted by a single $k$-substring complement starting at location $i$. Hence,
complementing the substring confined to coordinated $[i,\,i+k-1]$ only flips two derivative
coordinates, $d_i$ and $d_{i+k}$, and leaves the codestring unchanged otherwise.
\end{lemma}

\begin{proof}
Write $\bc' = \mathrm{c}_{k,i}(\bc)$, so $c'_t = c_t \xor 1$ for
$t \in [i,\,i+k-1]$ and $c'_t = c_t$ otherwise. Then, compare
$\Delta'_j = c'_{j-1}\xor c'_j$ with $\Delta_j = c_{j-1}\xor c_j$:
\begin{itemize}
  \item $j = i$: only $c_i$ is flipped, so
        $\Delta'_i = \Delta_i \xor 1$.
  \item $i < j < i+k$: both $c_{j-1}$ and $c_j$ are flipped, so their XOR is
        unchanged and $\Delta'_j = \Delta_j$.
  \item $j = i+k$: only $c_{i+k-1}$ is flipped, so
        $\Delta'_{i+k} = \Delta_{i+k} \xor 1$.
  \item all other $j$: neither neighbor is flipped, and hence $d'_j = d_j$.
\end{itemize}
Because we augmented the codestrings to length $n+1$ by setting $c_{n+1}=0$, the index $i+k$ lies in
$\{1,\dots,n+1\}$ even when the substring reaches the right end of the string. Hence, only $\Delta_i$ and $\Delta_{i+k}$ change.
\end{proof}

As a result, in the codeword derivative domain, a complement error equals two substitutions at a fixed
distance $k$. Since $i$ ranges over $n-k+1$ values, there are only $O(n)$
possible error patterns. These may be fixed as follows, using a shortened Hamming code.

Let $N=n+1$ and $r=\lceil\log(n+2)\rceil$, so that $2^{r}-1\ge n+1=N$. Choose a
primitive element of the field $\alpha\in\FF_{2^{r}}$ whose its multiplicative order is
$2^{r}-1$. As in the Hamming code construction, regard each field element as a column in $\FF_2^{\,r}$ and define
the parity-check matrix $H\in\FF_2^{\,r\times N}$ via its columns' binary encoding of the corresponding power of $\alpha,$
\[
  \bh_j \;=\; \alpha^{\,j}\in\FF_{2^{r}},\qquad j=1,\dots,N .
\]
Note that the $\bh_1,\dots,\bh_N$ are all distinct, nonzero elements of the underlying finite field. Let
\[
  \mathcal{C}\;=\;\bigl\{\,\bc\in\{0,1\}^n \;:\; H\,\Delta(\bc)=\zero\,\bigr\}.
\]
Clearly, this is a shortened Hamming/cyclic code imposed on the derivative of the binary strings.

\begin{proposition}
\label{prop:correct}
The shortened Hamming code $\mathcal{C}$ can corrects one
substring-complementation error of exact length $k$. Its redundancy is
\[
  n-\log|\mathcal{C}|\;\le\; r\;=\;\bigl\lceil\log(n+2)\bigr\rceil=\log n+O(1).
\]
\end{proposition}

\begin{proof}
Let $\bc\in\mathcal{C}$ be the transmitted codeword and let the channel corrupt the codeword via $\mathrm{C}_{k,i}$ where $i$ is unknown. Assume that the resulting binary string is $\by$. To decode $\by$, 
form the syndrome over $\FF_{2^{r}}$,
\[
  \sigma\;=\;H\,\Delta(\by).
\]
By Lemma~\ref{lem:deriv},
$\Delta(\by)=\Delta(\bc)\xor(\boldsymbol{\delta}_i\xor\boldsymbol{\delta}_{i+k})$,
and since $H\,\Delta(\bc)=\zero$ and $H$ is $\FF_2$-linear,
\[
  \sigma
  \;=\; \bh_i \xor \bh_{i+k}
  \;=\; \alpha^{\,i}+\alpha^{\,i+k}
  \;=\; \alpha^{\,i}\,\bigl(1+\alpha^{\,k}\bigr)
  \qquad\text{in } \FF_{2^{r}} .
\]

Now $1\le k\le n<2^{r}-1$, and $\alpha$ has order $2^{r}-1$, so $\alpha^{k}\neq1$
and the constant $\beta:=1+\alpha^{k}$ is a fixed nonzero element of the  field.
Hence:
\begin{itemize}
  \item If $\sigma=0$ the decoder declares no error.
  \item Otherwise $\alpha^{i}=\sigma\,\beta^{-1}$, and since
        $i\mapsto\alpha^{i}$ is injective on $\{1,\dots,n-k+1\}$, the position is
        recovered by a discrete logarithm,
        $i=\dlog_{\alpha}\!\bigl(\sigma\,\beta^{-1}\bigr)$. This is possible since $k$ is known a priori. The decoder then flips
        $y_i,\dots,y_{i+k-1}$, restoring $\bc$.
\end{itemize}
The key observation is as follows: distinct positions result in distinct syndromes because
$\alpha^{i}\beta=\alpha^{i'}\beta\Rightarrow\alpha^{i}=\alpha^{i'}\Rightarrow
i=i'$. Thus every single length-$k$ complementation error pattern is uniquely correctable. To establish the last claim, note that $\mathcal{C}$ is defined by the $r$ $\FF_2$-linear equations
$H\,\Delta(\bc)=\zero$, so $|\mathcal{C}|\ge 2^{\,n-r}$ and the redundancy is at most $r=\lceil\log(n+2)\rceil$. The redundancy of the code does not grow with $k$: the second boundary flip is at a known offset and carries no independent positional information. Algebraically, it only contributes the fixed scalar factor $1+\alpha^{k}$.
\end{proof}

\begin{remark} Both constructions using VT-like constraints asymptotically meet the $k$-burst lower bound in~\eqref{eq:k_burst_LB_combined} for large $n$.
As such, they match the performance of the cyclic codes described by Abramson and Abdel-Ghaffar et al.~\cite{abramson1959class,abdel1986existence}. According to the latter, an optimal code for $k=2$ can be generated using the polynomial $(1+x)p(x)$, where $p(x)$ is a primitive polynomial of degree $3$. A more restricted generator for special value of the degree provide optimal burst-correction codes for $k=3$, but most other results are purely existential.
\end{remark}

\section{Modified Fire-Codes for Permutation Error Correction} \label{sec:perm-fire}

The next simple and general construction uses the fact that permutation errors preserve the weight of the affected substring. We can add this constraint into the classical Fire code construction for arbitrary bursts of length $k$ and save one bit of redundancy as follows.
\begin{proposition}\label{prop:Fire_save_one_bit}
Let $n$ and $k$ be positive integers such that $n>2k$.
Let $\mc{C}$ be a (nontrivial) $k$-burst-correcting Fire code of block length $n$ and redundancy $r$ as defined in~\eqref{eq:Fire_shortened} in Appendix~\ref{app:Fire}.
Write $g(x)=(x^{2k-1}+1)f(x)$ for the associated generator of $\mc{C}$,
    and define $h_k(x)\coloneqq (x^{2k-1}+1)/(x+1)=1+x+x^2+\cdots+x^{2k-2}$.
Then, the code
\begin{align*}
    \mc{C}_{\tn{even}}\coloneqq \{c(x)\in\mbb{F}_2[x]~:~ \deg(c)<n,~\left.h_k(x)f(x)\middle|c(x)\right.\}
\end{align*}
    has block length $n$ and redundancy $r-1$ and can correct a single even-weight $k$-burst.
\end{proposition}
\begin{proof}
The block length of $\mc{C}_{\tn{even}}$ is clearly $n$ by construction.
At the same time,
    the assumption that $\mc{C}$ is nontrivial implies that $\deg(g)<n$,
    and thus $\deg(h_k f)=\deg(g)-1<n$ as well.
It follows that
    the redundancy of $\mc{C}_{\tn{even}}$ is $\deg(h_k f)=\deg(g)-1=r-1$.

It remains to prove that $\mc{C}_{\tn{even}}$ can correct any single even-weight $k$-burst.
Assume for contradiction that $\mc{C}_{\tn{even}}$ cannot correct an even-weight $k$-burst.
Then,
    there exist $c_1(x),c_2(x)\in\mc{C}_{\tn{even}}$ and two even-weight bursts $e_1(x)$ and $e_2(x)$ such that
    $c_1(x)\neq c_2(x)$ but
\begin{align}
    c_1(x)+e_1(x)=c_2(x)+e_2(x).\label{eq:even_weight_Fire}
\end{align}
Define $c_3(x)\coloneqq c_1(x)+c_2(x)$,
    which is also in $\mc{C}_{\tn{even}}$ by linearity of $\mc{C}_{\tn{even}}$,
    and in particular $c_3(x)\neq 0$.
Thus from \eqref{eq:even_weight_Fire}
    we have
\begin{align}
    c_3(x)+e_1(x)=e_2(x).\label{eq:even_weight_Fire_c3}
\end{align}
Since both $e_1(x)$ and $e_2(x)$ have even-weights,
    we have $x+1|e_1(x)$ and $x+1|e_2(x)$,
    and thus from \eqref{eq:even_weight_Fire_c3} we also have
\begin{align}
    x+1|c_3(x).\label{eq:even_weight_Fire_xp1}
\end{align}
At the same time,
    since $c_3(x)\in\mc{C}_{\tn{even}}$ we also have
\begin{align}
    h_k(x)f(x)|c_3(x).\label{eq:even_weight_Fire_hkf}
\end{align}
Observe that $h_k(1)=1$ and by the assumption on $f$ from Proposition~\ref{prop:Fire} we have $f(1)=1$ as well.
Therefore,
    $x+1$ does not divide $h_k(x)f(x)$,
    and thus combining \eqref{eq:even_weight_Fire_xp1} and \eqref{eq:even_weight_Fire_hkf} gives
\begin{align*}
    (x+1)h_k(x)f(x)|c_3(x).
\end{align*}
Recall that 
    $\deg (c_3)<n$ by the definition of $\mc{C}_{\tn{even}}$
    and that $(x+1)h_k(x)f(x)=g(x)$ by construction. Thus, 
    we have $c_3(x)\in \mc{C}$ as well. Finally, 
    an even-weight $k$-burst is itself a $k$ burst.
Together with the fact that $0\in\mc{C}$ and $c_3(x)\neq 0$,
    Equation \eqref{eq:even_weight_Fire_c3} implies that
    $\mc{C}$ cannot correct a $k$-burst,
    which contradicts the error-correction capability assumptions on $\mc{C}$.
\end{proof}
A direct consequence of Proposition \ref{prop:Fire_save_one_bit}
    is the following: together with the redundancy upper-bound on Fire codes in \eqref{eq:Fire_ub} in Appendix \ref{app:Fire},
    one can construct even-weight $k$-burst-correcting codes (and hence $k$-permutation-correcting codes)
    with block length $n$ and redundancy at most
\begin{align*}
    \max\{\log(n+1)+2k-1,3k-2\}.
\end{align*}

\section{A Refined Existential Results on Cyclic Permutation Codes} \label{sec:cyclic-refined}

In what follows, we prove an existential result on cyclic substring permutation codes for the special case where $k$ is prime.
\begin{theorem}\label{thm:cyc_perm_k_prime}
Let $k$ be a prime. There exists a length-$k$ cyclic substring permutation error-correcting code
with redundancy
\begin{align}
\log n + 20 \log k + O(1).
\label{eq:cyc_perm_k_prime}
\end{align}
\end{theorem}
\begin{remark}\label{rmk:thm_cyc_perm_k_prime}
Theorem \ref{thm:cyc_perm_k_prime} holds both when $k$ is a constant prime number and when $k$ grows with $n$ while maintaining its primality. However,
    it should be noted that
    the result in Theorem~\ref{thm:cyc_perm_k_prime} is practically relevant only when $k\leq n^{1/18}$,
    since for $k>n^{1/18},$
    Theorem~\ref{thm:CP_exact_k_red_UB}
    already implies the existence of a code with redundancy
    $2\log n + 2\log k + o(1) \leq \log n + 20 \log k + O(1)$.

\end{remark}
The proof of Theorem~\ref{thm:cyc_perm_k_prime} proceeds through several steps, outlined next.

We first propose the following lemma to confine the length-$k$ permutation error window into a segment of length $O(k^3)$.
\begin{lemma}\label{lem:locate_window_kcube}
Let $k$ be a prime number,
    $\boldx$ be a length-$n$ binary sequence, and
    $\boldy$ be obtained from $\boldx$ via a length-$k$ cyclic permutation.
Define
\begin{align*}
    g_1(\boldx)&=f_1(\boldx)\coloneqq \sum_{t=1}^n t x_t,\\
    g_2(\boldx)&\coloneqq \sum_{t=1}^n t^2 x_t,
\end{align*}
    and then define
    $\mathcal{D} g_i \coloneqq g_i(\boldy)-g_i(\boldx)$ for $i=1,2$. Also, define
\begin{align}
    L\coloneqq \frac{k(2k^2+3k+7)}{6}.
    \label{eq:L_def}
\end{align}
Then,
    given the knowledge of
    $\boldy$, 
    $\mathcal{D} g_1$, and $\mathcal{D} g_2$,
    we have the following:
\begin{itemize}
    \item
        If $\mathcal{D} g_1=0$,
            then $\boldx=\boldy$.
    \item 
        If $\mathcal{D} g_1\neq 0$,
            then $\boldx\neq \boldy$,
            and one can determine a substring in $\boldy$ of length $L$ 
            such that the length-$k$ circularly permuted substring 
            is \emph{contained} in the longer segment.
\end{itemize}
\end{lemma}
\begin{proof}
By the definition of a cyclic substring permutation error,
    there exists an $i\in[0,n-k]$ and a $s\in[1,k-1]$ such that
\begin{align*}
    \boldy = \boldx_{1:i} \circ \boldx_{i+s+1:i+k}\circ \boldx_{i+1:i+s} \circ \boldx_{i+k+1:n}.
\end{align*}
In words,
    the error model circularly permutes the substring comprising positions  $[i+1,i+k]$.
Define $\bolda \coloneqq \boldx_{i+1:i+s}$ and $\boldb \coloneqq \boldx_{i+s+1:i+k}$,
    so that we can write
\begin{align*}
    \boldx &= \boldx_{1:i} \circ \bolda \circ \boldb \circ \boldx_{i+k+1:n},\\
    \boldy &= \boldx_{1:i} \circ \boldb \circ \bolda \circ \boldx_{i+k+1:n}.
\end{align*}
A direct calculation shows that
\begin{align}
    \mathcal{D} g_1 = (k-s)A - sB,\label{eq:Delta_g_1}
\end{align}
    where $A$ and $B$ denote the Hamming weights of $\bolda$ and $\boldb$, respectively.
The reason is that for each $1$ in $\bolda$,
    transforming $\boldx$ into $\boldy$ shifts this $1$ by $|\boldb|=k-s$ positions to the right. A similar argument holds for $1$s in $\boldb$.
    
We now claim that 
    if $\mathcal{D} g_1=0$ then $\boldy=\boldx$.
Assume that $\mathcal{D} g_1= 0$.
Since $\mathcal{D} g_1= (k-s)A - sB $,
    we get
\begin{align}
    kA=s(A+B). \label{eq:Delta_g_1_zero}
\end{align}
In particular,
    from \eqref{eq:Delta_g_1_zero} we have
    $k|s(A+B)$.
Then,
    since $k$ is a prime number and $s\in[1,k-1]$,
    we have
    $k | A+B$.
Recall that the lengths of $\bolda$ and $\boldb$ sum to $k$,
    and thus $0\leq A+B\leq k$.
It follows that either $A+B=0$ or $A+B=k$.
Note that $A+B=0$ implies that both $\bolda$ and $\boldb$ are all-zero sequences,
    in which case
    switching $\bolda$ and $\boldb$ leads to the same sequence,
    i.e. $\boldy=\boldx$.
A similar argument holds for $A+B=k$,
    which corresponds to the case that both $\bolda$ and $\boldb$ are the all-one sequences.
This argument establishes the claim.

In the following we assume $\mathcal{D} g_1\neq 0$.
We now compute
\begin{align}
    \mathcal{D} g_2
        &= g_2(\boldy)-g_2(\boldx)\nonumber\\
        &= \sum_{t=1}^n t^2(y_t-x_t)\nonumber\\
        &= \sum_{t=i+1}^{i+k} t^2(y_t-x_t)\label{eq:same_outside}\\
        &= \sum_{j=1}^k (i+j)^2(y_{i+j}-x_{i+j})\label{eq:j_substitute}\\
        &= \sum_{j=1}^k (i^2+2ij+j^2)(y_{i+j}-x_{i+j})\nonumber\\
        &= i^2\sum_{j=1}^k (y_{i+j}-x_{i+j})+2i\sum_{j=1}^k j(y_{i+j}-x_{i+j})\nonumber\\
        &\qquad\qquad+\sum_{j=1}^k j^2(y_{i+j}-x_{i+j}),
        \label{eq:Delta_g_2_expand}
\end{align}
    where~\eqref{eq:same_outside} follows from the fact that $\boldx$ and $\boldy$ agrees on $[1,n]\setminus[i+1,i+k]$ and~\eqref{eq:j_substitute} follows from the change of variable $j=t-i$.
Then, we have
\begin{align}
    \sum_{j=1}^k(y_{i+j}-x_{i+j})=0.
    \label{eq:Delta_g_2_sum_to_0}
\end{align}
The claim holds because 
    $\boldx_{i+1:i+k}=\bolda\circ\boldb$ and $\boldy_{i+1:i+k}=\boldb\circ\bolda$,
    and thus $\boldx_{i+1:i+k}$ and $\boldy_{i+1:i+k}$ have the same Hamming weight.
Consequently,
\begin{align}
    \sum_{j=1}^k j(y_{i+j}-x_{i+j})
        &=\sum_{j=1}^k (i+j)(y_{i+j}-x_{i+j})\label{eq:Delta_g_2_term_2_add_0}\\
        &=\sum_{t=i+1}^{i+k} t(y_t-x_t)\label{eq:Delta_g_2_term_2_t_substitute}\\
        &=\sum_{t=1}^n t(y_t-x_t)\label{eq:Delta_g_2_same_outside}\\
        &=\mathcal{D} g_1,\label{eq:Delta_g_2_term_2_result}
\end{align}
    where \eqref{eq:Delta_g_2_term_2_add_0} is obtained by adding $0=i\sum_{j=1}^k(y_j-x_j)$ to both sides of the first equality,
    \eqref{eq:Delta_g_2_term_2_t_substitute} follows from the change of variable $t=i+j$,
    and \eqref{eq:Delta_g_2_same_outside} follows from the fact that $\boldx$ and $\boldy$ agree on $[1,n]\setminus[i+1,i+k]$.
Substituting \eqref{eq:Delta_g_2_sum_to_0} and \eqref{eq:Delta_g_2_term_2_result} into \eqref{eq:Delta_g_2_expand} leads to
\begin{align}
    \mathcal{D} g_2 = 2i\mathcal{D} g_1 + \sum_{j=1}^k j^2(y_{i+j}-x_{i+j}).
    \label{eq:Delta_g_2_simplified}
\end{align}
Since we have assumed that $\mathcal{D} g_1\neq 0$,
    we see from~\eqref{eq:Delta_g_2_simplified} that
\begin{align}
    i=  \frac{\mathcal{D} g_2}{2\mathcal{D} g_1}-\frac{1}{2\mathcal{D} g_1}\sum_{j=1}^k j^2(y_{i+j}-x_{i+j}).
    \label{eq:Delta_g_2_i_simplified}
\end{align}

Then,
    using the bounds $|\mathcal{D} g_1|\geq 1$ and $|x_{i+j}-y_{i+j}|\leq 1$,
    we get
\begin{align}
    \left|\frac{1}{2\mathcal{D} g_1}\sum_{j=1}^k j^2(y_{i+j}-x_{i+j})\right|
    &\leq \frac{1}{2}\sum_{j=1}^k j^2\nonumber\\
    &=\frac{k(k+1)(2k+1)}{12}.
    \label{eq:i_deviation_bound}
\end{align}
Putting \eqref{eq:i_deviation_bound} into \eqref{eq:Delta_g_2_i_simplified},
    we obtain that
    the index $i$ lies in the (real) 
    interval
\begin{align*}
    \left[ \frac{\mathcal{D} g_2}{2\mathcal{D} g_1}- \frac{k(k+1)(2k+1)}{12},\frac{\mathcal{D} g_2}{2\mathcal{D} g_1}+\frac{k(k+1)(2k+1)}{12}\right].
\end{align*}
Lastly,
    the permuted window $[i+1,i+k]$ is contained in the (real) interval
\begin{align*}
    &\left[ \frac{\mathcal{D} g_2}{2\mathcal{D} g_1}- \frac{k(k+1)(2k+1)}{12}+1,\right.\\
    &~\left.\frac{\mathcal{D} g_2}{2\mathcal{D} g_1}+\frac{k(k+1)(2k+1)}{12}+k\right],
\end{align*}
    which is of length
\begin{align*}
    \frac{k(k+1)(2k+1)}{6}+k-1=L-1,
\end{align*}
    where we used the definition of $L$ in \eqref{eq:L_def}. The claim then follows by noting that
    any length-$(L-1)$ real interval
    can contain at most $L$ consecutive integers.
\end{proof}

Lemma \ref{lem:locate_window_kcube}
    locates the permuted substring within another substring of length $L=O(k^3)$. The next lemma shows how to construct a ``local hash function'' for any length-$(2L)$ binary sequence that can help with further positioning of the cyclic shift.
\begin{lemma}\label{lem:greedy_coloring}
Define
\begin{align}
    M&\coloneqq  (k-1)^2(2L-k+1)^2+1,
    \label{eq:M_def}\\
    &=O(k^8)\nonumber,
\end{align}
    where $L=O(k^3)$ is defined in \eqref{eq:L_def}.
There exists a function
\begin{align*}
    h:\{0,1\}^{2L}\rightarrow [0,M-1]
\end{align*}
    satisfying the following property: for any distinct $\boldu,\boldv\in\{0,1\}^{2L}$
    such that $\boldu$ is obtained by applying two length-$k$ cyclic permutations to $\boldv$,
    we have $h(\boldu)\neq h(\boldv)$.
\end{lemma}
\begin{proof}
Consider an undirected, simple graph $G=(V,E)$ with vertex set $V=\{0,1\}^{2L}$ and edge set
\begin{align*}
E=\{&\{\boldu,\boldv\}~:~\boldu\neq\boldv\tn{ and }\boldu \tn{ is obtained by applying two }\nonumber\\
&\tn{length-}k \tn{ cyclic permutations to }\boldv\}.
\end{align*}
\begin{remark}
This graph $G$ is the ``confusion graph'' of the cyclic permutation error model with string input length $2L$ (i.e. $\{\boldu,\boldv\}$ is an edge if and only if $\boldu\neq \boldv$ and they both lead to the same output).
\end{remark}
The maximum degree of $G$,
    denoted by $D$,
    satisfies
\begin{align}
    D\leq (k-1)^2(2L-k+1)^2.\label{eq:D_ub}
\end{align}
The above inequality holds because there are $2L-k+1$ possible substring locations,
    and each substring has at most $k-1$ nontrivial circularly permuted outputs. It follows that 
    using a greedy coloring algorithm,
    there exists a coloring of $G$ using $D+1\leq M$ colors (see, e.g.\, \cite[Proposition 5.1.13]{west2001introduction}).
    
As the final step, we note that one can identify each color with an integer in $[0,M-1]$, and let the hash $h$ map each vertex to its color.
\end{proof}
We are now ready to prove Theorem \ref{thm:cyc_perm_k_prime}.
\begin{proof}
As outlined in Remark \ref{rmk:thm_cyc_perm_k_prime},
    it suffices to consider
    $k\leq n^{1/18}$.
In this case,
    $L=O(k^3)=O(n^{1/6})\ll n$. For each $\tau\in[1,\lfloor\frac{n}{L}\rfloor],$ define
\begin{align*}
    \boldx^{(\tau)}\coloneqq \boldx_{(\tau-1)L+1:(\tau+1)L},
\end{align*}
    where we assume that $\boldx^{(\lfloor\frac{n}{L}\rfloor)}$ is padded with zeros as needed, i.e., 
$\boldx^{(1)}=\boldx_{1:2L}$,
$\boldx^{(2)}=\boldx_{L+1:3L}$,
$\boldx^{(3)}=\boldx_{2L+1:4L}$. In particular,
    by construction,
    any length-$L$ substring in $\boldx$ is contained in at least one $\boldx^{(\tau)}$.

We now describe a code
    that can
    can correct a length-$k$ cyclic substring permutation error.
For each $c_1\in[0,2k^2-1]$, $c_2\in[0,4nk^2-1]$, $c_3\in[0,M-1]$, and $c_4\in[0,M-1]$, define
\begin{align*}
&\mc{C}_{c_1,c_2,c_3,c_4}\coloneqq
\{
\boldx\in\{0,1\}^n~:~\\
&g_1(\boldx)\equiv c_1~\bmod 2k^2,
g_2(\boldx)\equiv c_2~\bmod 4nk^2,\\
&\sum_{\tau \tn{ odd}} h(\boldx^{(\tau)})\equiv c_3~\bmod M,
\sum_{\tau \tn{ even}} h(\boldx^{(\tau)}) \equiv c_4~\bmod M
\}.
\end{align*}
Let $\boldx\in\mc{C}_{c_1,c_2,c_3,c_4}$
    and $\boldy$ be obtained by applying a length-$k$ cyclic permutation to $\boldx$.
We show next how to recover $\boldx$
    from $\boldy$ and the parameters $c_1,c_2,c_3,c_4$.

First,
    we show that we can recover $\mathcal{D} g_1$ from $\boldy$ and $c_1$, where we recall that $\mathcal{D} g_1=g_1(\boldy)-g_1(\boldx)$. From~\eqref{eq:Delta_g_1},
    since $A,B\in[0,k-1]$ and $s\in[1,k-1]$,
    we have
\begin{align*}
    |\mathcal{D} g_1|
    & \leq (k-s)A+sB\\
    & \leq (k-s)(k-1)+s(k-1)\\
    &=k(k-1).
\end{align*}
Therefore,
    $\mathcal{D} g_1\in[-k^2+k,k^2-k]$,
    which is an interval of length $2k^2-2k+1\leq 2k^2$.
It follows that knowing $\mathcal{D} g_1\bmod 2k^2$ allows us to uniquely recover $\mathcal{D} g_1$. We also have
\begin{align*}
    \mathcal{D} g_1 \equiv (g_1(\boldy)-c_1)\bmod 2k^2.
\end{align*}
This show that we can recover $\mathcal{D} g_1$ from $\boldy$ and $c_1$.

Similarly,
    from \eqref{eq:Delta_g_2_simplified},
    using $i\in[0,n-k]$, $|\mathcal{D} g_1|\leq k(k-1)\leq k^2$, and $\left|\sum_{j=1}^k j^2(y_{i+j}-x_{i+j})\right|\leq\frac{k(k+1)(2k+1)}{6}$,
    we have
\begin{align*}
    \left|\mathcal{D} g_2\right|
    &\leq 2(n-k)k^2+\frac{k(k+1)(2k+1)}{6}\\
    &= 2nk^2-2k^3+\frac{k(k+1)(2k+1)}{6}\\
    &\leq 2nk^2-1,
\end{align*}
    where the last inequality holds
    since $2k^3-\frac{k(k+1)(2k+1)}{6}\geq 1$.
By a similar argument,
    we can recover $\mathcal{D} g_2$ from $\mathcal{D} g_2 \bmod 4nk^2$,
    which equals $(g_2(\boldy)-c_2) \bmod 4nk^2$.

By Lemma \ref{lem:locate_window_kcube},
    if $\mathcal{D} g_1=0$ then we can recover $\boldx$ since 
    $\boldx=\boldy$.
In the following we assume $\mathcal{D} g_1\neq 0$.
By Lemma~\ref{lem:locate_window_kcube},
    we can find a substring of length $L$ in $\boldy$ that contains the circularly permuted substring of interest.
Let $\tau^{\star}\in[1,\lfloor\frac{n}{L}\rfloor]$ be such that the substring guaranteed by Lemma ~\ref{lem:locate_window_kcube} lies in $[(\tau^{\star}-1)L+1,(\tau^{\star}+1)L]$.
Here we assume that $\tau^{\star}$ is odd,
    and the case where $\tau^{\star}$ is even
    will follow from the same argument.
By construction,
    $\boldx$ and $\boldy$ agree on $[1,n]\setminus [(\tau^{\star}-1)L+1,(\tau^{\star}+1)L]$,
    and thus
\begin{align*}
    \sum_{\tau \tn{ odd}} \left(h(\boldy^{(\tau)})-h(\boldx^{(\tau)})\right)=h(\boldy^{(\tau^{\star})})-h(\boldx^{(\tau^{\star})}).
\end{align*}
Rearranging the terms and taking modulo $M$ leads to
\begin{align}
    h(\boldx^{(\tau^{\star})})\equiv
    \left(c_3-\sum_{\substack{\tau\tn{ odd}\\ \tau\neq\tau^{\star}}}h(\boldy^{(\tau)}) \right) \bmod M.
    \label{eq:h_x_tau_star}
\end{align}
Since $h(\boldx^{(\tau^{\star})})\in[0,M-1]$,
    Equation~\eqref{eq:h_x_tau_star} allows us to compute $h(\boldx^{(\tau^{\star})})$ from $c_3$ and $\boldy$.
    
We now perform a brute-force search by applying a length-$k$ cyclic permutation to $\boldy^{(\tau^{\star})}$ in all possible manners, and check if the output $\boldz^{(\tau^{\star})}$
    satisfies $h(\boldz^{(\tau^{\star})})=h(\boldx^{(\tau^{\star})})$.
By Lemma~\ref{lem:greedy_coloring},
    only the correct cyclic permutation that maps $\boldy^{(\tau^{\star})}$ back to $\boldx^{(\tau^{\star})}$ satisfies the constraint.
That is,
    any $\boldz^{(\tau^{\star})}$ with $\boldz^{(\tau^{\star})}\neq \boldx^{(\tau^{\star})}$
    satisfies $h(\boldz^{(\tau^{\star})}) \neq h(\boldx^{(\tau^{\star})})$.
We can thus uniquely recover $\boldx^{(\tau^{\star})}$,
    and 
    since $\boldy$ and $\boldx$ agree on $[1,n]\setminus [(\tau^{\star}-1)L+1,(\tau^{\star}+1)L]$,
    we can recover $\boldx$. These arguments jointly show that
    $\mathcal{C}_{c_1,c_2,c_3,c_4}$
    is a length-$k$ cyclic substring permutation error-correcting code.

As a final step, by the pigeonhole principle,
    there exist
    $c_1^{\star}\in[0,2k^2-1]$,
    $c_2^{\star}\in[0,4nk^2-1]$,
    $c_3^{\star}\in[0,M-1]$, and
    $c_4^{\star}\in[0,M-1]$ 
    such that
\begin{align*}
    \left|\mathcal{C}_{c_1^{\star},c_2^{\star},c_3^{\star},c_4^{\star}}\right|\geq \frac{2^n}{(2k^2)(4nk^2)M^2}=\frac{2^n}{8nk^4M^2}.
\end{align*}
It follows that the redundancy of $\mathcal{C}_{c_1^{\star},c_2^{\star},c_3^{\star},c_4^{\star}}$ is at most
\begin{align*}
    \log n + 4\log k + 2 \log M + 3
    = \log n + 20\log k + O(1),
\end{align*}
    where we used the fact that $M=\frac{4}{9}k^8+O(k^7)$ from the definition~\eqref{eq:M_def}.
\end{proof}

\section{A $\poly(n,k)$-time code for cyclic permutations} \label{sec:cyclic-efficient}

The above described construction relies on a coloring algorithm and exhaustive search which make its complexity of encoding/decoding unacceptably high. By adding $k$ bits of redundancy,
    we can modify the code in Theorem~\ref{thm:cyc_perm_k_prime} to obtain a code with $\log n + k + O(\log k)$ redundancy and $\poly(n,k)$ time complexity. Our construction is restricted to prime $k$ of properly bounded growth-rate.

\begin{theorem}\label{thm:cyc_perm_k_prime_poly_time}
Let $k$ be a prime number such that $2L\leq n$,
    where $L$ is defined as in~\eqref{eq:L_def}.
In particular, set
    $k=O(n^{1/3})$.
    
Then, there exists a length-$k$ cyclic permutation error-correcting code
    with redundancy 
\begin{align*}
    \log n + k + 49\log k + O(\log \log k)
\end{align*}
such that both the encoder and decoder have time complexity $\poly(n,k)$.
\end{theorem}
The redundancy $\log n + k + 49\log k$ 
    in Theorem~\ref{thm:cyc_perm_k_prime_poly_time} deviates even further from the lower bound $\log n + \log k$ established in Theorem~\ref{thm:CP_exact_k_red_LB},
    but it is still \emph{smaller than the redundancy of Fire codes} characterized by~\eqref{eq:Fire_lb}, 
    provided both $n$ and $k$ are large enough.
    
Recall that in the proof of Theorem~\ref{thm:cyc_perm_k_prime}
    we use the greedy coloring function (hash) $h$ described in Lemma~\ref{lem:greedy_coloring},
    which takes $2^{O(k^3)}$ time to compute.
To address the complexity issue, inspired by the syndrome compression technique~\cite{sima2020syndrome}, 
    we first propose an efficient coloring algorithm for the graph $G$ defined in Lemma~\ref{lem:greedy_coloring},
    which is facilitated by the following Lemma. 

\begin{lemma}\label{lem:compressed_coloring}
Let $k$ be a positive integer and $L$ be as defined in \eqref{eq:L_def}.
There exist $\widehat{M}=O(k^{22}\log k)$ and a function
\begin{align*}
    \widehat{h}:\{0,1\}^{2L}\rightarrow [0,\widehat{M}-1]
\end{align*}
    that satisfy the following property: for any distinct $\boldu,\boldv\in\{0,1\}^{2L}$
    such that $\boldu$ is obtained by applying two length-$k$ cyclic permutations to $\boldv$,
    we have $\widehat{h}(\boldu)\neq \widehat{h}(\boldv)$.
Furthermore,
    $\widehat{h}$ can be computed in $\mathrm{poly}(k)$ time.
\end{lemma}
\begin{proof} Let $G$ be the graph defined in the proof of Lemma \ref{lem:greedy_coloring}.
It suffices to find a coloring of $G$ with $\widehat{M}$ colors that is computable in $\poly(k)$ time.

For each $\boldu\in\{0,1\}^{2L}$,
    let $\val(\boldu)$
    be the integer whose binary representation is $\boldu$.
Then,
    define for each $\boldu\in\{0,1\}^{2L}$ that
\begin{align}
  a(\boldu) &\coloneqq \min\{a \tn{ prime}~:~ \val(\boldu)\not\equiv\val(\boldv)\bmod a  \nonumber\\ &\qquad\qquad\tn{ for all }\boldv\in N(\boldu)\,\},
  \label{eq:a_def}
\end{align}
    where $N(\boldu)$ denotes the neighborhood of $\boldu$ in $G$, and by convention, we define $\min\emptyset\coloneqq \infty$.
We claim that $a$ has the following property.
\begin{claim}\label{clm:a_ub}
Define $D^*\coloneqq (k-1)^2(2L-k+1)^2$,
    which is an upper bound on the the maximum degree of $G$ as shown in~\eqref{eq:D_ub}.
For each $\boldu\in\{0,1\}^{2L}$,
    we have
    $a(\boldu)\leq p_{2LD^*}$,
    where $p_j$ denotes the $j$th prime number.
In particular, we have $a(\boldu)<\infty$.
\end{claim}
\begin{proof}
By contradiction, assume that $a(\boldu)>p_{2LD^*},$ (including the case $a(\boldu)=\infty$). That is,
    for each $j\in[1,2LD^*]$ we have
\begin{align}
    \val(\boldu)\equiv \val(\boldv) \bmod p_j \tn{ for some }\boldv\in N(\boldu).
    \label{eq:clm_a_ub_contradict}
\end{align}

    Consider the integer
\begin{align}
    P\coloneqq \prod_{\boldv\in N(\boldu)}|\val(\boldv)-\val(\boldu)|,\label{eq:P_def}
\end{align}
    which is strictly positive since $\val(\cdot)$ is injective.
Then,
    note that we have $|\val(\boldv)-\val(\boldu)|< 2^{2L}$,
    since $0\leq \val(\boldv),\val(\boldu)\leq 2^{2L}-1$
Therefore, from~\eqref{eq:P_def} we have
\begin{align}
    P< 2^{2LD^*}.
    \label{eq:P_ub}
\end{align}

At the same time,
    from \eqref{eq:clm_a_ub_contradict} and \eqref{eq:P_def} we have for each $j\in[1,2LD^*]$ that
    $p_j|P$.
    This leads to
\begin{align}
    \prod_{j=1}^{2L{D^*}}p_j|P.
    \label{eq:prod_pj_divides_P}
\end{align}
Then,
    using the bound $p_j\geq 2$ and the fact that $P$ is strictly positive, from~\eqref{eq:prod_pj_divides_P} we obtain
\begin{align}
    P\geq 2^{2LD^*}.\label{eq:P_lb}
\end{align}
Comparing \eqref{eq:P_ub} and \eqref{eq:P_lb} leads to a contradiction.
\end{proof}

We are now ready to describe the constructio of the desired function $\widehat{h}$ as
\begin{align}
    \widehat{h}(\boldu)\coloneqq \eta(a(\boldu))p_{2LD^*}+\left(\val(\boldu)\bmod a(\boldu)\right),
    \label{eq:widehath_def}
\end{align}
    where $\eta(p_j)\coloneqq j-1$ for each prime $p_j$.
The range of $\widehat{h}$ is bounded as follows:
first,
    by Claim~\ref{clm:a_ub} we know that
    $\eta((a(\boldu))\in[0,2LD^*-1]$.
Furthermore,
    by Claim~\ref{clm:a_ub} we also know that
    $\val(\boldu)\bmod a(\boldu)\in[0,a(\boldu)-1]\subseteq[0,p_{2LD^*}-1]$.
Using the bound $p_j\leq j(\ln j + \ln\ln j)$ for $j\geq 6$ \cite[Corollary of Theorem 3]{rosser1962approximate},
    we obtain $p_{2LD^*}= O(LD^*\log(LD^*))$.
Therefore,
    $\widehat{h}(\boldu)$ is at most
\begin{align*}
    2LD^*p_{2LD^*}
    &=O(L^2{D^*}^2\log LD^*)\\
    &=O(k^{22}\log k),
\end{align*}
    where we used the properties that $L=O(k^3)$ from \eqref{eq:L_def} and $D^*=O(k^2L^2)=O(k^8)$.
At the same time,
    it is clear that $\widehat{h}(\boldu)\geq 0$ by construction.
Therefore,
    we can choose $\widehat{M}\coloneqq 2LD^*p_{2LD^*}=O(k^{22}\log k)$.

We show next that $\widehat{h}$ is indeed a coloring of $G$.
Assume for contradiction that there exist a $\boldu\in\{0,1\}^{2L}$ and a $\boldv\in N(\boldu)$
    such that $\widehat{h}(\boldu)=\widehat{h}(\boldv)$.
By the argument above,
    we know that 
    each $\widehat{h}(\boldu)$ is uniquely associated with an ordered pair $\left(\eta(a(\boldu)),\val(\boldu)\bmod a(\boldu)\right)$.
Therefore,
    we have
    both 
\begin{align}
    \eta(a(\boldu))=\eta(a(\boldv))\label{eq:equal_left}
\end{align}
    and
\begin{align}
    \left(\val(\boldu)\bmod a(\boldu)\right) = \left(\val(\boldv)\bmod a(\boldv)\right).\label{eq:equal_right}
\end{align}
Note that \eqref{eq:equal_left} 
    implies that $a(\boldu)=a(\boldv)$,
    and thus from~\eqref{eq:equal_right} we have
\begin{align*}
    \val(\boldu)\equiv\val(\boldv)\bmod a(\boldu),
\end{align*}
    which contradicts the definition of $a(\boldu)$ in~\eqref{eq:a_def}.
These arguments establish that
    $\widehat{h}$ is indeed a coloring of $G$.

To conclude,
    we analyze the time complexity of $\widehat{h}$.
The computation of $\widehat{h}$ proceeds through the following steps:
\begin{enumerate}
    \item Generate the list of the first $2LD^*$ prime.
    By the bound from~\cite{rosser1962approximate},
        it suffices to list all the primes up to $A\coloneqq 2LD^*(\ln(2LD^*)+\ln\ln(2LD^*)$,
        which takes $O(A\cdot\polylog(A))=\poly(k)$ time using the sieve of Eratosthenes~\cite{sorenson1990introduction}.
    \item Enumerate $\boldv\in N(\boldu)$, which takes $D^*=O(k^8)$ time.
    \item
        For each prime $p$ in the generated list
            and each neighbor $\boldv$ of $\boldu$,
            test whether $\val(\boldu) \equiv \val(\boldv) \mod p$, which takes $\poly(k)$ time (assuming dividing an $2L$-bit integer with another $2L$-bit integer 
            takes $\poly(2L)=\poly(k)$ time). This step produces the desired $a(\boldu)$ and $\eta(a(\boldu))$.
    \item Compute $\widehat{h}(\boldu)$ based on \eqref{eq:widehath_def}, which takes $\poly(\log k)$ time.
\end{enumerate}
Since each step can be executed in $\poly(k)$ time,
    the overall computation of $\widehat{h}(\boldu)$ requires $\poly(k)$ time.
\end{proof}
We are now ready to prove Theorem \ref{thm:cyc_perm_k_prime_poly_time}.
\begin{proof}
Define
\begin{align*}
    g_0(\boldx)\coloneqq \sum_{t=1}^n \boldx_t.
\end{align*}
Furthermore, let $n'\coloneqq n+k$. For each $c_0\in[0,k-1]$, $c_1\in[0,2k^2-1]$, and $c_2\in[0,4n'k^2-1]$,
    define
\begin{align*}
    \mathcal{C}_{c_0,c_1,c_2}\coloneqq\{&\boldx\in\{0,1\}^n~:~g_0(\boldx)\equiv c_0\bmod k,\\
    &g_1(\boldx)\equiv c_1\bmod 2k^2,~g_2(\boldx)\equiv c_2\bmod 4n'k^2\}.
\end{align*}
By the pigeonhole principle, there exist $c_0^{\star}\in[0,k-1]$, $c_1^{\star}\in[0,2k^2-1]$, and $c_2^{\star}\in[0,4n'k^2-1]$
    such that
\begin{align*}
    |\mathcal{C}_{c_0^{\star},c_1^{\star},c_2^{\star}}|\geq \frac{2^n}{k\cdot 2k^2\cdot 4n'k^2}=\frac{2^n}{8n'k^5}.
\end{align*}
In other words,
    the redundancy of $\mathcal{C}_{c_0^{\star},c_1^{\star},c_2^{\star}}$ is at most
    $\log (8n'k^5)=\log(n+k)+5\log k + 3$.

We show next how to find
    $c_0^{\star}$,$c_1^{\star}$, and $c_2^{\star}$ using dynamic programming.
For each  $c_0\in[0,k-1]$, $c_1\in[0,2k^2-1]$, $c_2\in[0,4n'k^2-1]$,
    and $t\in[1,n]$,
    define
\begin{align*}
    N_{c_0,c_1,c_2,t}\coloneqq 
    &\left|\{\boldy\in\{0,1\}^t~:~ 
    g_0(\boldy)\equiv c_0\bmod k,\right.\\
     &\left. g_1(\boldy)\equiv c_1\bmod 2k^2,~g_2(\boldy)\equiv c_2\bmod 4n'k^2
    \}\right|.
\end{align*}
Also, for all $c_0\in[0,k-1]$, $c_1\in[0,2k^2-1]$, and $c_2\in[0,4n'k^2-1]$ define
\begin{align*}
    N_{c_0,c_1,c_2,0}\coloneqq \begin{cases}
        1, \tn{ if }c_0=c_1=c_2=0,\\
        0, \tn{ otherwise}.
    \end{cases}
\end{align*}
Then,
    it can be seen that
    for each $t\in[1,n]$,
    we have
    the following recursion:
\begin{align}
    N_{c_0,c_1,c_2,t}=N_{c_0,c_1,c_2,t-1}+N_{c_0',c_1',c_2',t-1},
    \label{eq:N_recursive}
\end{align}
    where $c_0'\coloneqq c_0-1\mod k$,
    $c_1'\coloneqq c_1-t\mod 2k^2$, and
    $c_2'\coloneqq c_2-t^2\mod 4n'k^2$.
Here the first term in \eqref{eq:N_recursive} corresponds to the case where the the $t$th bit is $0$,
    while the second term corresponds to the case where the the $t$th bit is $1$.
By~\eqref{eq:N_recursive},
    one can compute a table consisting of all $N_{c_0,c_1,c_2,n}$ for $c_0\in[0,k-1]$, $c_1\in[0,2k^2-1]$, $c_2\in[0,4n'k^2-1]$. This enumeration procedure will also count towards the overall complexity -- see the complexity analysis in the next paragraph.
Then,
    $c_0^{\star}$, $c_1^{\star}$, and $c_2^{\star}$
    can be chosen to be the first maximizer of $N_{c_0,c_1,c_2,n}$ with respect to a prescribed lexicographic order of the tuples $(c_0,c_1,c_2)$.
    
We now analyze the time complexity of finding $c_0^{\star}$, $c_1^{\star}$, and $c_2^{\star}$.
Note that each $N_{c_0,c_1,c_2,t}$ obeys an upper bound
    $N_{c_0,c_1,c_2,t}\leq 2^t\leq 2^n$,
    and thus every $N_{c_0,c_1,c_2,t}$ can be represented by an $(n+1)$-bit integer.
It follows that
    each computation of \eqref{eq:N_recursive} takes $\mathrm{ADD}_{n+1}$,
    where $\mathrm{ADD}_{n+1}$ denotes the time complexity of adding two $(n+1)$-bit integers.
Together with the fact that \eqref{eq:N_recursive} is performed $k\cdot 2k^2\cdot 4n'k^2 \cdot n=8n(n+k)k^5$ times, the computation of the table consisting of all $N_{c_0,c_1,c_2,n}$
    takes $8n(n+k)k^5\mathrm{ADD}_{n+1}$ time.
Then,
    a brute-force comparison over all $c_0\in[0,k-1]$, $c_1\in[0,2k^2-1]$, $c_2\in[0,4n'k^2-1]$ to find $c_0^{\star}$, $c_1^{\star}$, and $c_2^{\star}$
    takes $8(n+k)k^5\mathrm{COMP}_{n+1}$ time,
    where $\mathrm{COMP}_{n+1}$ denotes the time complexity to compare two $(n+1)$-bit integers.
Therefore,
    the overall procedure takes $8n(n+k)k^5\mathrm{ADD}_{n+1}+8(n+k)k^5\mathrm{COMP}_{n+1}$ time.
Assuming $\mathrm{ADD}_{n+1}=\poly(n)$ and $\mathrm{COMP}_{n+1}=\poly(n)$,
    the overall time complexity is thus $\poly(n,k)$.

To efficiently convert between $\{0,1\}^{n-(\log (n+k) + 5\log k + 3)}$ and ${\mc{C}}_{c_0^{\star},c_1^{\star},c_2^{\star}}$,
    we can construct a deterministic finite-state automaton (DFA)
    such that the set of length-$n$ binary strings accepted by this DFA is exactly ${\mc{C}}_{c_0^{\star},c_1^{\star},c_2^{\star}}$.
See, e.g., \cite[Appendix B]{chen2026correctingcontextualdeletionsdna} for more details.
To be more precise,
    for a binary-alphabet DFA with state space $Q$,
    its length-$n$ accepted strings can be enumerated 
    by the ``unranking'' function $\mathrm{unrank}$,
    whose inverse function is
    the ``ranking'' function $\mathrm{rank}$.
Both $\mathrm{rank}$ and $\mathrm{unrank}$ can be computed in $O(T_{\tn{construct}}T_{\tn{transition}}|Q|n\mathrm{ADD}_n)$ time,
    where $T_{\tn{construct}}$ is the time to construct the DFA,
    $T_{\tn{transition}}$ is the time to calculate the state transition,
    and $\mathrm{ADD}_n$ once again denotes the time it takes to add two $n$-bit integers.

We construct the desired DFA,
    which is based on the recursive formula in~\eqref{eq:N_recursive}, as follows.
The state space is $Q=\{\epsilon\}\cup \left([0,k-1]\times[0,2k^2-1]\times[0,4n'k^2-1]\times[1,n] \right)$,
    where $\epsilon$ denotes the empty string.
Each state other than $\epsilon$
    is identified as a four-tuple $(c_0,c_1,c_2,t)$.
The initial state is $\epsilon$.
The state transitions $\delta$ are defined as follows:
\begin{itemize}
    \item $\delta(\epsilon,0)=(0,0,0,1)$ and $\delta(\epsilon,1)=(1,1,1,1)$.
    \item 
        For each $t\in[1,n-1]$, $c_0\in[0,k-1]$,
        $c_1\in[0,2k^2-1]$, and $c_2\in[0,4n'k^2-1]$,
            define
            $\delta((c_0,c_1,c_2,t),0)=(c_0,c_1,c_2,t+1)$ and
            $\delta((c_0,c_1,c_2,t),1)=(c_0',c_1',c_2',t+1)$,
            where $c_0'\coloneqq c_0+1\mod k$,
    $c_1'\coloneqq c_1+t+1\mod 2k^2$, and
    $c_2'\coloneqq c_2+(t+1)^2\mod 4n'k^2$.
    \item 
        The state transitions from the states of the form $(c_0,c_1,c_2,n)$ can be arbitrary.
    \item 
        The accepted state is $(c_0^{\star},c_1^{\star},c_2^{\star},n)$.
\end{itemize}
This DFA has $|Q|=8nn'k^5+1$ states,
    and a length-$n$ binary sequence is accepted if and only if it is in $\mathcal{C}_{c_0^{\star},c_1^{\star},c_2^{\star}}$.
In addition,
    it can be checked that $T_{\tn{construct}}$ and $T_{\tn{transition}}$ are both $\poly(n,k)$.
It follows that
    both the ranking function $\mathrm{rank}:\mathcal{C}_{c_0^{\star},c_1^{\star},c_2^{\star}}\rightarrow [0,|\mathcal{C}_{c_0^{\star},c_1^{\star},c_2^{\star}}|-1]$
    and the unranking function $\mathrm{unrank}:[0,|\mathcal{C}_{c_0^{\star},c_1^{\star},c_2^{\star}}|-1]\rightarrow \mathcal{C}_{c_0^{\star},c_1^{\star},c_2^{\star}}$
    can be computed in $\poly(n,k)$ time.

We are now ready to describe the overall encoding procedure. For each $\boldx\in\{0,1\}^{n-(\log n' + 5\log k + 3)}$,
    identify $\boldx$ with an integer in $[0,\frac{2^n}{8n'k^5}-1]$ and
    let $\overline{\boldx}\coloneqq\unrank(\boldx)\in\{0,1\}^n$.
Then,
    the encoded words is
\begin{align*}
    \Enc(\boldx)\coloneqq \overline{\boldx}\circ 0^k \circ \widehat{H}_{\tn{odd}}(\overline{\boldx}\circ 0^k)\circ \widehat{H}_{\tn{even}}(\overline{\boldx}\circ 0^k),
\end{align*}
    where for any binary sequence $\boldz$ we let
\begin{align*}
    \widehat{H}_{\tn{odd}}(\boldz)&\coloneqq \left(\sum_{\tau \tn{ odd}} \widehat{h}(\boldz^{(\tau)})\right) \bmod \widehat{M}\\
    \widehat{H}_{\tn{even}}(\boldz)&\coloneqq \left(\sum_{\tau \tn{ even}} \widehat{h}(\boldz^{(\tau)})\right) \bmod \widehat{M}.
\end{align*}
Note that we have assumed 
    $2L\leq n$
    so that $\widehat{H}_{\tn{odd}}$ and $\widehat{H}_{\tn{even}}$ are nonempty sums.

The length of $\Enc(\boldx)$ is 
\begin{align*}
    |\Enc(\boldx)|
    &=n+k+2\log \widehat{M}\\
    &=n+k+44\log k + O(\log \log k),
\end{align*}
    and thus the overall redundancy equals 
    $|\Enc(\boldx)|-|\boldx|=\log (n+k) + k + 49\log k + O(\log \log k)=\log n + k + 49\log k + O(\log \log k)$,
    where in the last step we used the fact that $k=O(n^{1/3})$ and thus $\log(n+k)=\log n +O(1)$.
    
Note that $\Enc(\boldx)$ can be computed in $\poly(n,k)$ time
    since both $\unrank$ and $\widehat{h}$ are computable in $\poly(n,k)$ time.

Before describing the decoding procedure,
    we first explain the high-level idea behind it. Write $\bolds_{\tn{info}}\coloneqq \widehat{H}_{\tn{odd}}(\overline{\boldx}\circ 0^k)\circ\widehat{H}_{\tn{even}}(\overline{\boldx}\circ 0^k)$, so that $\Enc(\boldx)$ can be expressed as
\begin{align*}
\Enc(\boldx)=\overline{\boldx}\circ 0^k \circ \bolds_{\tn{info}}.
\end{align*}
We classify the length-$k$ cyclic permutation errors into the following types:
\begin{itemize}
    \item Type I: The length-$k$ substring lies within the $\overline{\boldx}$ part of $\Enc(\boldx)$ (i.e. the length-$k$ substring $[i+1,i+k]$ satisfies $[i+1,i+k]\subseteq[1,n]$).
    \item Type II: The length-$k$ substring straddles $\overline{\boldx}$ and $0^k$ (i.e. $[i+1,i+k]\cap[1,n]\neq\emptyset$ and $[i+1,i+k]\cap[n+1,n+k]\neq\emptyset$).
    \item Type III: The length-$k$ substring straddles $0^k$ and $\bolds_{\tn{info}}$.
    \item Type IV: The length-$k$ substring lies within the $\bolds_{\tn{info}}$ part.
\end{itemize}
Note that we have excluded the case where the length-$k$ substring aligns with the $0^k$ part of $\Enc(\boldx)$,
    since $0^k$ remains unchanged after a length-$k$ cyclic permutation.

Next, observe that $g_1(\overline{\boldx})=g_1(\overline{\boldx}\circ 0^k)\equiv c_1^{\star} \bmod 2k^2$ and $g_2(\overline{\boldx})=g_2(\overline{\boldx}\circ 0^k)\equiv c_2^{\star} \bmod 4n'k^2$. Therefore,
    for Type I and Type II errors,
    we can use $c_1^{\star}$, $c_2^{\star}$, and $\bolds_{\tn{info}}$ and apply the decoding procedure described in the proof of Theorem \ref{thm:cyc_perm_k_prime} to recover $\overline{\boldx}\circ 0^k$,
    which in turn recovers $\overline{\boldx}$.
On the other hand,
    for Type III and Type IV errors,
    the $\overline{\boldx}$ part is intact
    and thus can be recovered directly. 

We now explain how to determine which type of error has occurred.
Suppose $\boldy$ is obtained from $\Enc(\boldx)$ via a length-$k$ cyclic substring permutation.
We consider two cases:
\begin{enumerate}
    \item The case $\boldy_{n+1:n+k}\neq 0^k$.
        In this case, a cyclic permutation error must have  happened,
            and in particular
            the error is either of Type II or Type III.
        We then examine the quantity $\mathcal{D} g_0\coloneqq g_0(\boldy_{1:n+k})-g_0(\overline{\boldx}\circ 0^k)$,
            which satisfies $\mathcal{D} g_0\equiv g_0(\boldy_{1:n+k})-c_0^{\star}\bmod k$. Note that for a Type II error we have
        $\mathcal{D} g_0=0,$ since cyclic permutations are weight-preserving.
        On the other hand,
            for a Type III error,
            $\mathcal{D} g_0=w(\boldy_{n+1:n+k})>0$
            since $\boldy_{1:n}=\overline{\boldx}$
            and $\boldy_{n+1:n+k}\neq 0^k$.
        Furthermore,
            since the weight of $\boldy_{n+1:n+k}$ is  ``borrowed'' from the $\bolds_{\tn{info}}$ part,
            we have the upper bound $w(\boldy_{n+1:n+k})\leq \lfloor k/2 \rfloor\leq k-1$.
        Based on these observations,
            we can perform decoding based on the value of $\mathcal{D} g_0\bmod k$:
        \begin{enumerate}
            \item If $\mathcal{D} g_0\equiv 0 \bmod k$,
                then we know that a Type II error has occurred. In particular,
                the $\bolds_{\tn{info}}$ part is intact,
                and thus as outlined above,
                we can use $c_1^{\star},c_2^{\star}$, and $\bolds_{\tn{info}}$ to recover $\overline{\boldx}\circ 0^k$ and then $\overline{\boldx}$.
            \item 
                If $\mathcal{D} g_0\not\equiv 0 \bmod k$,
                    then we know that a Type III error has occurred. In this case,
                    we simply recover $\overline{\boldx}$ from the first $n$ bits of $\boldy$.
        \end{enumerate}
    \item 
        The case $\boldy_{n+1:n+k}= 0^k$.
        Note that in this case,
            Type I and Type II errors in $\Enc(\boldx)$
            correspond to a cyclic permutation in $\overline{\boldx}\circ 0^k$,
            while
            Type III and Type IV errors in $\Enc(\boldx)$ do not cause errors in $\overline{\boldx}\circ 0^k$.
        Therefore,
            following the proof of Lemma  \ref{lem:locate_window_kcube},
            we examine the quantity 
            $\mathcal{D} g_1\coloneqq g_1(\boldy_{1:n+k})-g_1(\overline{\boldx}\circ 0^k)$,
            which satisfies $\mathcal{D} g_1\equiv g_1(\boldy_{1:n+k})-c_1^{\star}\bmod 2k^2$.
        We then consider the following cases:
        \begin{enumerate}
            \item
                If $\mathcal{D} g_1=0$,
                    then there is no error confined to the first $n+k$ bits of $\boldy$,
                    and we can directly recover $\overline{\boldx}=\boldy_{1:n}$.
            \item 
                If $\mathcal{D} g_1\neq 0$,
                    then either a Type I or a Type II error has occurred in the first $n+k$ bits of $\boldy$.
                In particular,
                    the $\bolds_{\tn{info}}$ part remains  intact,
                    and we can recover $\overline{\boldx}$ in a similar way.
        \end{enumerate}
\end{enumerate}
It can be checked that
    the decoding procedure in all of the cases above runs in $\poly(n,k)$ time. Recall that in the decoding procedure described in the proof of Theorem \ref{thm:cyc_perm_k_prime},
    after we determine a substring of length $L=O(k^3)$ as described in Lemma~\ref{lem:locate_window_kcube},
    we locate $\boldy^{(\tau^*)}$ of length $2L$, brute-force enumerate all the possible length-$k$ cyclic-permutation outcomes $\boldz^{(\tau^*)}$ and test
    whether $h(\boldz^{(\tau^*)})=h(\boldx^{(\tau^*)})$.
In this theorem $h$ is replaced with $\widehat{h}$,
   the computation of which runs in $\poly(k)$ time,
    and there are at most $(2L-k+1)k=O(k^4)$ possible outcomes from $\boldy^{(\tau^*)}$,
    so that the overall decoding procedure indeed runs in $\poly(n,k)$ time. As a final step, we remark that after recovering $\overline{\boldx}$,
    we can recover the transmitted message
    by $\boldx=\rank(\overline{\boldx})$,
    which requires $\poly(n,k)$ time.
\end{proof}

\section{Multiple substring permutation bursts} \label{sec:multiple-bursts}

To address the issue of correcting multiple substring permutation errors, we make the obvious observation that a code correcting all even-weight bursts 
corrects all weight-preserving bursts for every codeword. 

We start with some auxiliary definitions. The integral of a binary string $\Int:\{0,1\}^n\to\{0,1\}^n$ is the running prefix sum
\[
  \Int(\bx)_j \coloneqq \oplus_{r=1}^{j} x_r, \qquad j=1,\dots,n .
\]
We also recall the definition of the derivative of a binary sequence,
$\Delta:\{0,1\}^n\to\{0,1\}^n$ (which in Definition~\ref{def:diff} had an additional boundary condition here omitted), 
\begin{align}
  \Delta(\bx)_j \coloneqq x_{j-1}\xor x_j, \qquad j=1,\dots,n,\quad x_0:=0 .
  \label{eq:Delta_same_length_def}
\end{align}

The following proposition is self-evident.
\begin{proposition}
\label{prop:inverse}
$\Delta$ and $\Int$ are mutually inverse, in the sense that
$\Delta(\Int(\boldx))=\Int(\Delta(\boldx))=\boldx$
    for all $\boldx\in\{0,1\}^n$.
In addition,
    both $\Delta$ and $\Int$ are $\FF_2$-linear.
\end{proposition}

Our proposed construction relies on the following easy-to-see observation.
Let $\be$ be an even burst of length at most $k$ 
starting at position $i$. Then $\Int(\be)$
is a burst of length at most $k-1$, starting at the same position $i$. To be more precise, we have
$\Int(\be)_i=1$ and $\supp\bigl(\Int(\be)\bigr)\subseteq[i,\,i+k-2]$.
To see this,
    note that
    the even-weight property of $\bolde$ ensures that $\Int(\bolde)_{i+k-1}=\oplus_{j=i}^{i+k-1}(\bolde)_j=0$.

Next, fix a code $\mathcal{C}'\subseteq\{0,1\}^n$ that can correct $t$ nonoverlapping bursts separated by at least one bit (so as to prevent two or more bursts forming a longer consecutive bust),
each confined to $k-1$ consecutive positions. Define
\[
  \mathcal{C}
  \coloneqq
  \Delta(\mathcal{C}')
  \;=\;\bigl\{\,\bc\in\{0,1\}^n \;:\; \Int(\bc)\in\mathcal{C}'\,\bigr\}.
\]
Encoding amounts to choosing a base codeword $\bx\in\mathcal{C}'$, and transmitting 
$\bc=\Delta(\bx)$. Decoding involves processing the received $\by$ by forming $\Int(\by)$, running the decoder for 
$\mathcal{C}'$ to obtain $\bx$, and outputting $\bc=\Delta(\bx)$. This property is formally characterized by the following proposition.
\begin{proposition}
\label{prop:t_k_burst}
Let the channel add $\be=\oplus_{s=1}^{t}\be_s$, where each $\be_s$ is an 
even-weight burst
of length 
$\ell_s \leq k$ starting
at position $i_s$, 
    and that the bursts are nonoverlapping
    in the sense that $[i_s,i_s+\ell_s-1]\cap [i_{s'},i_{s'}+\ell_{s'}-1]=\emptyset$ for all $s\neq s'$.
Then the above construction recovers $\bc$ uniquely. Moreover, since $\Int$ and $\Delta$
are bijections, the codes have the same redundancy.
That is,
\[
  r(\mathcal{C})=r(\mathcal{C}').
\]
\end{proposition}
\begin{proof}
Let $\bc\in\mathcal{C}$, so $\bx:=\Int(\bc)\in\mathcal{C}'$. By linearity of $\Int$,
\[
  \Int(\by)\;=\;\Int(\bc)\xor\Int(\be)
             \;=\;\bx\;\xor\;(\oplus_{s=1}^{t}\Int(\be_s)).
\]
Note that as argued above, each $\Int(\be_s)$ is a burst of length at most $k-1$ starting
at position $i_s$.
Furthermore,
    since the intervals $\{[i_s,i_s+\ell_s-1]\}_{1\leq s \leq t}$
    are nonoverlapping,
    the intervals $\{[i_s,i_s+\ell_s-2]\}_{1\leq s \leq t}$
    are separated one bit apart from each other.
Hence,
    $\Int(\bolde)$ consists of bursts of length at most $k-1$  which are at least one bit apart. Since $\Int(\boldc)=\boldx\in\mc{C}'$, 
    the decoder of
$\mathcal{C}'$ thus returns $\bx$ from $\Int(\by)$,    
and one can recover $\bc=\Delta(\bx)$ by Proposition ~\ref{prop:inverse}.

Also, the fact that $|\mathcal{C}|=|\mathcal{C}'|$ follows because $\Delta$ is a
bijection, ensuring that the codes have equal redundancy.
\end{proof}

We demonstrate next a simple application of Proposition~\ref{prop:t_k_burst} for permutation error correction.
Note that 
for the special case of a single even-weight burst of length at most a constant $k$,
one can combine Proposition~\ref{prop:t_k_burst} with 
the existential results in~\cite{abdel1986existence} for $(k-1)$-burst-correcting codes asymptotically meeting the necessary lower bound in~\eqref{eq:k_burst_LB_combined}.
This procedure establishes the existence of optimal even-weight $k$-burst error-correcting codes, which can also correct length-$k$ permutation errors.

On one hand,
    the optimal $(k-1)$-burst correcting code in \cite{abdel1986existence} has redundancy asymptotically $\log n + (k-1)-1+o(1)=\log n + k - 2 +o(1)$.
On the other hand,
    we derive a necessary redundancy lower bound for any codes correcting a single even-weight $k$-burst error.
To count the number of such errors, the left endpoint is fixed to $1$, while the $k-1$ ``following'' bits have to have an odd weight. There are $2^{k-2}$ such substrings per starting position and $n-k+1$ starting positions, and thus there are $(n-k+1)2^{k-2}$ even-weight bursts of length $k$ overall.
Correcting such a burst therefore requires at least $\log n+k-2+o(1)$ redundant bits. Consequently,
    combining Proposition~\ref{prop:t_k_burst} with the asymptotically optimal codes in~\cite{abdel1986existence} leads to asymptotically optimal single even-weight $k$-burst-correcting codes. We also consider another special case where Proposition \ref{prop:t_k_burst} is applicable. Over $\Ftwo$ a weight-preserving burst on a substring of length $2$ must flip both
positions of the substring whose content is $01$ or $10$, that is, it reduces to an adjacent transposition. Then, $\Acc$ maps the pattern $11$ at positions $(j, j+1)$ to a single $1$ at
position $j$. The code design problem reduces to the statement that correcting $t$ adjacent transpositions is equivalent, via the integral precoder, to correcting $t$ single substitution errors. Using BCH codes results in redundancy $t \log\, (n+1),$ which suffices, 
instead of the redundancy $2t\log\,(n+1)$ required for correcting $2t$ errors representing $t$ bursts of length two.

\appendices

\section{A Review of Fire Codes and their  Redundancy}\label{app:Fire}
We first review the definition and the error-correction guarantees of Fire codes.
\begin{proposition}\label{prop:Fire}
\cite[Corollary of Theorem 8.14]{mceliece2002theory}
Let $k$ be a positive integer.
Let $f(x)$ be an irreducible polynomial over $\mbb{F}_2[x]$ of degree $m\geq k$ such that
    $f(0)\neq 0$ and
    $f(x)$ does not divide $x^{2k-1}+1$.
Write $n_0$ for the period of $f(x)$
    (i.e.\ the smallest positive integer $n_0'$ such that $f(x)$ divides $x^{n_0'}+1$).
Define
\begin{align*}
    g(x)\coloneqq(x^{2k-1}+1)f(x).
\end{align*}
Then,
    the cyclic code generated by $g(x)$ has block length $N=\lcm(2k-1,n_0)$
    and redundancy $m+2k-1$,
    and it is a $k$-burst-correcting code.
This code is called the \emph{Fire code}.
\end{proposition}
To emphasize the dependency on the burst length $k$,
    we also call the code in Proposition~\ref{prop:Fire}
    the \emph{$k$-burst-correcting Fire code}.
    
Now we establish the redundancy bounds on $k$-burst-correcting Fire codes of arbitrary block length $n$,
    which is obtained by constructing a $k$-burst-correcting Fire code in Proposition~\ref{prop:Fire} of block length $N\geq n$ and then shortening it. More precisely,
    if $f(x)$ defines a valid $k$-burst-correcting Fire code of block length $N\geq n$ in the sense of Proposition~\ref{prop:Fire}, then
\begin{align}
\mc{C}\coloneqq \{c(x)\in\mbb{F}_2[x]~:~ \deg(c)<n,~ g(x)|c(x)\}\label{eq:Fire_shortened}
\end{align}
    is the associated shortened code.
In particular, the code in \eqref{eq:Fire_shortened} is still a $k$-burst-correcting code,
    and we call it a $k$-burst-correcting Fire code of block length $n$.
Note that by Reiger's bound from~\eqref{eq:Reiger}
    the code in \eqref{eq:Fire_shortened} is non-trivial only when $n>2k$.
We then focus on the regime $n> 2k$.
\begin{proposition}\label{prop:Fire_red}
Let $n$ and $k$ be two positive integers
    such that $n> 2k$.
We have the following:
\begin{itemize}
    \item There exists a $k$-burst-correcting Fire code of block length $n$ 
    and redundancy at most 
    \begin{align}
        \max\{\log (n+1)+2k,3k-1\}.
        \label{eq:Fire_ub}
    \end{align}
    \item Any nontrivial $k$-burst-correcting Fire code of block-length $n$ has redundancy at least
    \begin{align}
        \max\{\log n+2k-\log (2k-1) - 1,3k-1\}.
        \label{eq:Fire_lb}
    \end{align}
\end{itemize}
\end{proposition}
\begin{proof}
We first prove the redundancy upper bound in~\eqref{eq:Fire_ub}.
Choose $m=\max\left\{\lceil\log(n+1)\rceil,k\right\}$
    and let $f(x)$ be a primitive polynomial in $\mbb{F}_2[x]$ of degree $m$.
The period of $f(x)$ is then $n_0=2^m-1\geq n$,
    and in particular $f(x)$ does not divide $x^{2k-1}+1$ since $2k-1<n\leq n_0$.
Together with the fact that $m\geq k$, the chosen $m$ and $f$ satisfy the assumptions in Proposition~\ref{prop:Fire}.
Therefore,
    the associated Fire code is well-defined,
    and it has block length
    $N=\lcm(2k-1,n_0)\geq n_0\geq n$
    and redundancy
\begin{align*}
    m+2k-1&\leq \max\{\log(n+1)+1,k\}+2k-1\\
    &=\max\{\log (n+1) + 2k ,3k-1\}.
\end{align*}
It follows that the associated shortened code has redundancy $\min\{n,m+2k-1\}\leq m+2k-1\leq \max\{\log (n+1) + 2k ,3k-1\}$ as well\footnote[1]{Note that the case redundancy equals $n$ corresponds to a trivial code (i.e.\ the code with only one codeword), which happens only when $n$ is too small.}.

We now prove the redundancy lower bound in~\eqref{eq:Fire_lb}. Since the shortened length-$n$ code is nontrivial, so is any original length-$N$ code before shortening. By the redundancy formula $m+2k-1$,
    it suffices to show that $m\geq \max\{\log n - \log(2k-1),k\}$.
By requirement, we already have $m\geq k$.
Also, from the formula of block length $N=\lcm(2k-1,n_0)$,
    we have
\begin{align*}
    n_0&= \gcd(2k-1,n_0)\frac{N}{2k-1}\\
    &\geq \frac{n}{2k-1},
\end{align*}
    since $N\geq n$ by assumption and $\gcd(2k-1,n_0)\geq 1$.
At the same time,
    since $f$ has degree $m$,
    is irreducible,
    and $f(0)\neq 0$,
    it can be shown that $f(x)|x^{2^m-1}-1$.
It then follows from the minimality of $n_0$ that
\begin{align*}
    2^m-1\geq n_0\geq \frac{n}{2k-1},
\end{align*}
    which implies $m\geq \log n -\log (2k-1)$.
\end{proof}

\section{Lower bound on $\binom{2k-2}{k-2}$}\label{app:binom_lb}
From \cite[Lemma 17.5.1]{cover1991elements},
     we have for each positive even integer $L$ that
\begin{align*}
    \binom{L}{L/2}\geq \frac{2^L}{\sqrt{2L}}.
\end{align*}
Setting $L=2k-2$ yields
\begin{align*}
    \binom{2k-2}{k-1}\geq \frac{2^{2k-2}}{2\sqrt{(k-1)}}.
\end{align*}
Then,
    from the identity
\begin{align*}
\binom{2k-2}{k-2}=\frac{k-1}{k}\binom{2k-2}{k-1},
\end{align*}
    we get
\begin{align*}
    \binom{2k-2}{k-2}&\geq \frac{k-1}{k}\frac{2^{2k-2}}{2\sqrt{(k-1)}}\\
    &= \sqrt{1-\frac{1}{k}}\frac{2^{2k-\textcolor{ColorYPChen}{3}}}{\sqrt{k}}\\
    & \geq \frac{2^{2k-\textcolor{ColorYPChen}{7/2}}}{\sqrt{k}},
\end{align*}
    where the last inequality follows from the fact that $1-\frac{1}{k}\geq \frac{1}{2}$ for $k\geq 2$.

\section{Asymptotic order of $d_{k-1}$}\label{app:dkm1_asymp}

In this appendix,
    we show that
\begin{align}
    d_{k-1}=\Theta(\frac{2^k}{k}),\label{eq:d_km1_asymp_goal}
\end{align}
    where 
    $d_{k-1}=\max_{s\in[0,k-1]}\left|\binom{k-1}{s}-\binom{k-1}{s-1}\right|$
    is defined in \eqref{eq:d_km1_def}.
First,
    by symmetry of the binomial coefficient,
    we can maximize over $s\leq \lfloor\frac{k-1}{2}\rfloor$.
That is,
    we have
\begin{align*}
    d_{k-1}
    &=\max_{s\in[0,\lfloor\frac{k-1}{2}\rfloor]}\left|\binom{k-1}{s}-\binom{k-1}{s-1}\right|\\
    &=\max_{s\in[0,\lfloor\frac{k-1}{2}\rfloor]}\left(\binom{k-1}{s}-\binom{k-1}{s-1}\right)\\
    &=\max_{s\in[0,\lfloor\frac{k-1}{2}\rfloor]}\frac{k-2s}{k}\binom{k}{s}.
\end{align*}
Then,
    for each $s\in[0,\lfloor\frac{k-1}{2}\rfloor]$,
    define
\begin{align*}
    d^{(s)}_{k-1}\coloneqq \frac{k-2s}{k}\binom{k}{s}.
\end{align*}
Next,
    for each $s\in[0,\lfloor\frac{k-1}{2}\rfloor-1]$
    we calculate
\begin{align*}
    \frac{d^{(s+1)}_{k-1}}{d^{(s)}_{k-1}}&=\frac{\frac{k-2s-2}{k}\binom{k}{s+1}}{\frac{k-2s}{k}\binom{k}{s}}\\
    &=\frac{(k-2s-2)(k-s)}{(k-2s)(s+1)}.
\end{align*}

We now have the following equivalent conditions for $s\in[0,\lfloor\frac{k-1}{2}\rfloor-1]$:
\begin{align*}
    &\frac{d^{(s+1)}_{k-1}}{d^{(s)}_{k-1}}\geq 1,\\
    \Leftrightarrow~&(k-2s-2)(k-s)\geq (k-2s)(s+1),\\
    \Leftrightarrow~&4s^2+(-4k+4)s+(k^2-3k)\geq 0,\\
    \Leftrightarrow~& s\leq \frac{k-1-\sqrt{k+1}}{2} \tn{ or } s\geq \frac{k-1+\sqrt{k+1}}{2}.
\end{align*}
However,
    since we have assumed $s\leq \lfloor\frac{k-1}{2}\rfloor-1$,
    the condition $s\geq \frac{k-1+\sqrt{k+1}}{2}$ cannot hold.
Therefore,
    we have $\frac{d^{(s+1)}_{k-1}}{d^{(s)}_{k-1}}\geq 1$ if and only if $s\leq \frac{k-1-\sqrt{k+1}}{2}$.
It follows that
\begin{align*}
    d_{k-1}=\frac{k-2s^{\star}}{k}\binom{k}{s^{\star}},
\end{align*}
    where the (possibly not unique) maximizer $s^{\star}$ is
\begin{align*}
    s^{\star}&\coloneqq \left\lfloor\frac{k-1-\sqrt{k+1}}{2}\right\rfloor+1\\
    &=\frac{k-\sqrt{k}}{2}+O(1).
\end{align*}

Finally,
    by \cite[(5.38), (5.43)]{spencer2014asymptopia},
    it can be shown that
\begin{align*}
    \binom{k}{s^{\star}}=\Theta(\binom{k}{k/2})=\Theta(\frac{2^k}{\sqrt{k}}).
\end{align*}
Together with the fact that 
$\frac{k-2s^{\star}}{k}=\frac{\sqrt{k}+O(1)}{k}=\Theta(\frac{1}{\sqrt{k}})$,
    we get \eqref{eq:d_km1_asymp_goal}.
    


\bibliographystyle{IEEEbib.bst}
\bibliography{refs}

\end{document}